\documentclass[lettersize,journal]{IEEEtran}
\usepackage{amsmath,amsfonts}
\usepackage{algorithmic}
\usepackage{algorithm}
\usepackage{array}
\usepackage{textcomp}
\usepackage{stfloats}
\usepackage{url}
\usepackage{verbatim}
\usepackage{graphicx}
\usepackage{cite}

\usepackage[ruled,lined,linesnumbered,algo2e]{algorithm2e}
\usepackage{subcaption}
\usepackage{amsthm}
\theoremstyle{plain}
\newtheorem{theorem}{Theorem}

\newtheorem{definition}{Definition}

\newtheorem{requirement}{Requirement}
\newtheorem{lemma}{Lemma}
\newtheorem{remark}{Remark}

\begin{document}

\title{Guaranteeing Data Privacy in Federated Unlearning with Dynamic User Participation}

\author{Ziyao Liu, Yu Jiang, Weifeng Jiang, Jiale Guo, Jun Zhao, and Kwok-Yan Lam
\thanks{Ziyao Liu, Yu Jiang, Weifeng Jiang, Jiale Guo, Jun Zhao, and Kwok-Yan Lam are with Nanyang Technological University, Singapore.  
E-mail: liuziyao@ntu.edu.sg, \{yu012, weifeng001\}@e.ntu.edu.sg, \{jiale.guo, junzhao, kwokyan.lam\}@ntu.edu.sg.}
\thanks{Manuscript received June 1, 2024; revised August 16, 2024.}
}

\markboth{Journal of \LaTeX\ Class Files,~Vol.~14, No.~8, August~2021}%
{Shell \MakeLowercase{\textit{et al.}}: A Sample Article Using IEEEtran.cls for IEEE Journals}

\IEEEpubid{0000--0000/00\$00.00~\copyright~2021 IEEE}

\maketitle

\begin{abstract}
Federated Unlearning (FU) is gaining prominence for its capability to eliminate influences of specific users' data from trained global Federated Learning (FL) models. A straightforward FU method involves removing the unlearned user-specified data and subsequently obtaining a new global FL model from scratch with all remaining user data, a process that unfortunately leads to considerable overhead. To enhance unlearning efficiency, a widely adopted strategy employs clustering, dividing FL users into clusters, with each cluster maintaining its own FL model. The final inference is then determined by aggregating the majority vote from the inferences of these sub-models. This method confines unlearning processes to individual clusters for removing the training data of a particular user, thereby enhancing unlearning efficiency by eliminating the need for participation from all remaining user data. However, current clustering-based FU schemes mainly concentrate on refining clustering to boost unlearning efficiency but without addressing the issue of the potential information leakage from FL users' gradients, a privacy concern that has been extensively studied. Typically, integrating secure aggregation (SecAgg) schemes within each cluster can facilitate a privacy-preserving FU. Nevertheless, crafting a clustering methodology that seamlessly incorporates SecAgg schemes is challenging, particularly in scenarios involving adversarial users and dynamic users. In this connection, we systematically explore the integration of SecAgg protocols within the most widely used federated unlearning scheme, which is based on clustering, to establish a privacy-preserving FU framework, aimed at ensuring privacy while effectively managing dynamic user participation. Comprehensive theoretical assessments and experimental results show that our proposed scheme achieves comparable unlearning effectiveness, alongside offering improved privacy protection and resilience in the face of varying user participation.
\end{abstract}


\section{Introduction}
\IEEEPARstart{I}{n} recent years, the concept of knowledge removal in AI systems has gained substantial attention from both academia and industry. This focus addresses a range of concerns, including data quality and sensitivity, copyright restrictions, obsolescence, and adherence to global privacy regulations such as GDPR \cite{regulation2018general}, introducing provisions like the ``Right To Be Forgotten" (RTBF), which allows individuals to request the removal of their personal data and its impact from trained Machine Learning (ML) models. Consequently, Machine Unlearning (MU) \cite{bourtoule2021machine,huynh2024fast,hu2024learn,liu2024threats,chen2021machine,han2024towards,hu2024duty,nguyen2022survey,liu2024towards,hu2023eraser,wang2023machine} has emerged as a pivotal mechanism in this context.

A straightforward MU approach involves retraining, which entails discarding the existing model and retraining a new one from scratch, excluding the unlearned data. However, this strategy can lead to significant computational costs, rendering it impractical for situations involving large-scale models or extensive datasets. To address these challenges, a variety of efficient MU techniques have recently been developed. These techniques are categorized into exact or approximate unlearning \cite{xuh2023machine}. Exact unlearning often involves retraining but narrows the focus to a subset of the data, thereby enhancing efficiency. On the other hand, approximate unlearning adjusts the parameters of the trained model to produce an unlearned model that simulates the outcome of retraining from scratch, offering a balance between efficiency and effectiveness.

Expanding on the foundational concepts of MU, Federated Unlearning (FU) \cite{cao2023fedrecover,liu2022right,ding2023strategic,liu2023survey,tao2024communication,wang2024server,jiang2024towards,guo2023fast,romandini2024federated,qiu2023fedcio,wang2023bfu,yuan2024towards} has been developed to address the challenge of data erasure within federated learning (FL) environments. In FL systems, multiple users train machine learning models locally, and these models are aggregated to form a global model. Then the server distributes this updated global model back to all users for further training in the next round of FL. This cycle repeats until the global model achieves convergence. Within this framework, the aim of FU is to allow the FL model to eliminate the influence of data from a specific FL user or any identifiable information tied to a user's partial dataset. Similar to the retraining-based approach for MU, a straightforward FU method involves removing the unlearned users\footnote{
In this work, we use the terms ``unlearned users" and ``target users" interchangeably.} and retraining a new global FL model with the remaining users. However, retraining in FU requires the participation of all FL users, leading to significant communication and computation costs. 

\IEEEpubidadjcol

To improve efficiency beyond the naive retraining-based FU approach, a commonly used strategy involves clustering. As depicted in Figure \ref{fig:clustered-fu}, this method divides FL users into clusters, with each cluster maintaining its own FL model. The final inference is then achieved by aggregating the majority vote from the inferences of these sub-models of clusters. As such, upon receiving a request to remove an FL user, only the cluster containing that user needs to conduct unlearning by retraining from scratch using their data following an FL process. This approach restricts the retraining process to individual clusters, thereby enhancing unlearning efficiency by eliminating the need for participation from all remaining users.

\begin{figure}[htbp!]
    \centering    
    \includegraphics[width=0.9\linewidth]{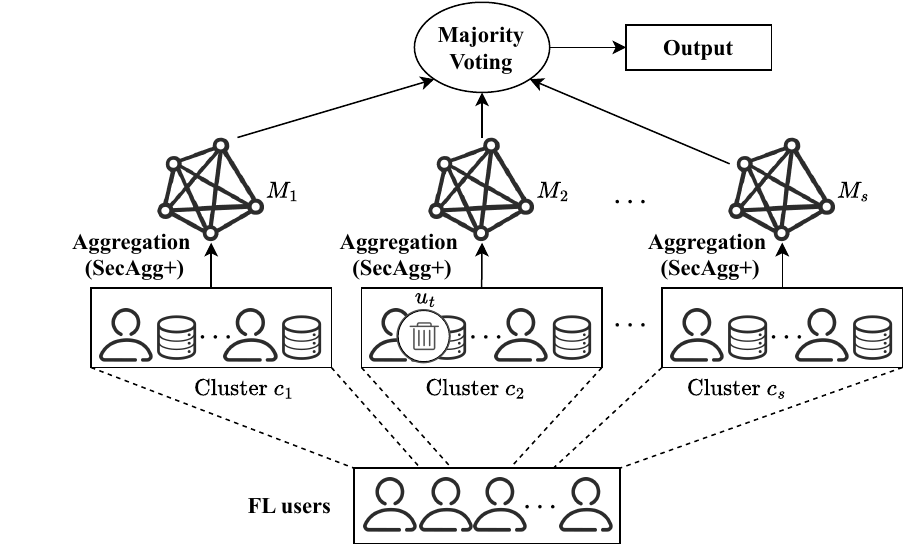}
    \caption{An illustrative example of clustering-based FU adapted from \cite{bourtoule2021machine}. FL users are divided into clusters. During the inference phase, test data is input into the model of each cluster, and the inferences from all clusters are aggregated to produce the final results based on majority voting. If an unlearning request is initiated to remove the target user $u_t$ in the cluster $c_2$, only the remaining users in cluster $c_2$ need to conduct the unlearning process, i.e., retraining from scratch following an FL style within the cluster. Users in other clusters can continue their FL training process or conduct the unlearning process according to the unlearning requests initiated in their respective clusters.}
    \label{fig:clustered-fu}
\end{figure}

However, severe information leakage can be explored through gradients, a risk that has been thoroughly examined in FL settings. As highlighted in \cite{zhu2019deep}, during the aggregation process of FL training, the exposure of users' locally trained models or gradients could result in significant information leakage. This vulnerability may enable an attacker to reconstruct images with pixel-level accuracy and texts with exact token matching. Therefore, privacy-enhancing techniques such as those based on Multi-Party Computation (MPC), Differential Privacy (DP), and Homomorphic Encryption (HE) are adopted to preserve the privacy of gradients in real-life FL applications. Among these techniques, privacy-preserving aggregation protocols \cite{bonawitz2017practical,bell2020secure,ma2023flamingo,so2023securing,bell2023acorn,liu2022efficient,liu2022privacy} are the most efficient due to their lower computation and communication costs compared to pure MPC and HE solutions, as well as their higher model usability compared to DP-based approaches. Privacy-preserving aggregation protocols enable the aggregation of users' gradients, ensuring that individual gradients remain private while allowing the server to obtain the aggregated results. Among those protocols, the family of secure aggregation (SecAgg) methods based on pair-wise masking \cite{bonawitz2017practical,liu2024dynamic,bell2020secure,bell2023acorn,ma2023flamingo,liu2023long} are widely adopted due to their efficiency and capability to handle dropout FL users.

Similar to FL systems, the aforementioned privacy risks stemming from information leakage through gradients also persist within each cluster in clustering-based FU schemes. Typically, integrating SecAgg schemes within each cluster can promote a privacy-preserving FU. However, developing a clustering methodology that seamlessly incorporates SecAgg schemes into the FU schemes poses significant challenges, especially in scenarios involving adversarial and dynamic users who drop out and must be removed from the cluster due to unlearning requests. The reason is that the security guarantees provided by Secure Aggregation (SecAgg) schemes are affected by the number of honest and malicious users within each cluster. Since both dropout users and unlearned users impact the cluster size, the security guarantees may be compromised. Additionally, as explained later, the state-of-the-art SecAgg protocol SecAgg+ \cite{bell2020secure} adopted in our proposed scheme involves using an $m$-regular graph as the communication topology to reduce communication overhead and Shamir secret sharing schemes to handle dropout users, which necessitates more specific considerations of compatibility. For more details on how SecAgg protocols work, we refer readers to Section \ref{sec:secagg}.

\textbf{Challenges of adopting SecAgg+ protocols.} Therefore, to ensure the adoption of the SecAgg+ protocol into clustering-based FU, it is crucial to guarantee that the conditions for using Shamir secret sharing schemes and $m$-regular graphs are satisfied within each cluster. Nevertheless, achieving such clustering requires not only consideration of the construction of $m$-regular graphs but also addressing potential privacy risks posed by various users, including
\begin{itemize}
    \item Adversarial users: users attempting to compromise the privacy of honest users.
    \item Dropout users: users who drop out of the systems.
    \item Unlearned users: users removed in response to requests.
\end{itemize}

\textit{Examples of privacy risks from clustering.} There exists a potential risk that an excessive proportion of the cluster could consist of adversarial users, specifically greater than $t$, which could potentially compromise the security guarantees provided by the Shamir secret sharing schemes. Likewise, the presence of dropout users and unlearned users can also impact the size of each cluster. Consequently, mechanisms to manage these users play a crucial role in maintaining the security guarantees offered by the SecAgg+ protocols.

Therefore, we aim to explore the following research questions to ensure the privacy guarantees of the SecAgg+ protocols within clustering-based FU schemes, taking the aforementioned factors into account:

\begin{description}
    \item[RQ1:] How should users be clustered to fulfill the requirements necessary for adopting SecAgg+ protocols within a clustering-based FU scheme?
    \item[RQ2:] What factors should be considered to maintain the privacy guarantees when managing dynamic users, especially considering dropout users and unlearned users?
\end{description}

\textbf{Our contributions.} In this connection, we introduce a tailored clustering method with supporting protocols, aimed at establishing a privacy-preserving FU scheme capable of providing security against adversarial users and dynamic users. First, we assess the security prerequisites for incorporating SecAgg+ protocols within clustering-based FU architectures. Subsequently, grounded in these security considerations, we develop a clustering algorithm tailored to meet these requirements, taking into account diverse user roles and the constructions of $m$-regular graphs in SecAgg+ protocols. Furthermore, we investigate potential threats posed by unlearned users in clustering-based FU schemes and suggest strategies to mitigate these risks. In summary, our key contributions are outlined as follows.

\begin{enumerate}
    \item We propose a federated unlearning scheme that guarantees user data privacy.
    \item We present a clustering algorithm specifically tailored for the adoption of the SecAgg+ protocol within clustering-based schemes, along with strategies to handle dynamic users.
    \item We explore the potential threats posed by unlearned users within clustering-based schemes and suggest strategies to mitigate these risks.
    \item We present the analysis of the impacts of different parameters on the trade-off between privacy guarantees, unlearning performance, and protocol efficiency, from both theoretical and experimental perspectives.
\end{enumerate}

\textbf{Connections and relationships to existing FU and FL works.}
We should note that our proposed scheme aims to integrate SecAgg protocols within clusters by only bounding the cluster size to provide additional security guarantees, as explained later in Section \ref{sec:requirements_on_clustering}. Therefore, our proposed method is independent of other existing clustering-based FU approaches that aim to improve efficiency or convergence rate, such as those based on data distribution \cite{qiu2023fedcio}, computational resources and model sparsity \cite{su2023asynchronous}, and data heterogeneity \cite{wang2023fedcsa}. Additionally, our proposed method can be combined with existing clustering methods by simply introducing bounds on cluster size. Similarly, our proposed method is independent of the unlearning and FL training process, i.e., how the unlearning algorithm is conducted within each cluster. We use the most straightforward retraining-based approach as an example to illustrate how to integrate SecAgg protocols, but unlearning can be conducted by any other more efficient means. Therefore, handling issues within FL systems, such as imbalanced situations, e.g., imbalanced unlearning requests and non-IID data, and resource constraints in terms of communication and computation, are orthogonal techniques that can be combined with our approach to further enhance the overall performance within the FU systems. In addition, we note that our proposed approach falls under privacy-enhancing schemes rather than providing security guarantees against threats such as backdoor attacks \cite{gao2024backdoor,gong2022backdoor,gong2023b3}, and those raised by malicious unlearning requests \cite{hu2024duty,qian2023towards,di2022hidden}. The detection and mitigation methods for these security risks are also orthogonal techniques that can be integrated with our proposed schemes. Since our proposed scheme provides privacy guarantees by integrating existing SecAgg protocols with only bounding the cluster size, its simplicity makes it attractive both for implementation and for further improvements through the integration of other techniques.


\section{Preliminaries and Notations}
\label{sec:preliminaries}

This section provides a brief overview of the preliminaries, including federated learning and unlearning, hypergeometric distribution, FL model convergence, and the definition of notations used in our proposed protocol.

\subsection{Federated Learning \& Unlearning}

The participants involved in federated learning can be categorized into two categories: (i) a set of $N$ users denoted as $U=\{u_1,u_2,\dots,u_N\}$, where each user $u_i \in U$ possesses its local dataset $\mathcal{D}_i$, and (ii) a central server represented as $S$. A typical FL scheme operates by iteratively performing the following steps until training is stopped: (a) Local model training: at round $n$, each FL user $u_i$ trains its local model based on a global model $M^n$ using its local dataset $\mathcal{D}_i$ to obtain gradients ${G}_i^n$. (b) Model uploading: each user $u_i$ uploads its gradients ${G}_i^n$ to the central server $S$. (c) Model aggregation: the central server $S$ collects and aggregates users' models $\{{G}_i^n\}^N$ with an aggregation rule, e.g., FedAvg \cite{mcmahan2017communication}, to update the global model ${M^{n+1}}$. (d) Model distribution: the central server $S$ distributes the updated global model ${M^{n+1}}$ to all FL users. Building upon the core principles of machine unlearning \cite{bourtoule2021machine} and the concept of RTBF, federated unlearning aims to enable the global FL model to remove the impact of an FL user or identifiable information associated with the partial data of an FL user, while preserving the privacy preservation of the decentralized learning process.

\subsection{Secure Aggregation}
\label{sec:secagg}
As previously discussed, integrating privacy-preserving aggregation protocols with the clustering-based FU scheme allows the central server to aggregate users' gradients within each cluster in a privacy-preserving manner. Among those protocols, the widely adopted SecAgg protocols\footnote{We omit some detailed constructions of SecAgg protocols for simplicity. Interested readers are encouraged to refer to \cite{bonawitz2017practical,bell2023acorn,bell2020secure} for a more comprehensive understanding of these details.} employ pairwise additive masking and Shamir's secret sharing \cite{shamir1979share} to facilitate efficient aggregation and ensure resilience against dropout users. Specifically, assuming there are $N$ users in an FL system, where user $u_i$ holds a vector of gradients $\boldsymbol{x}_i$, SecAgg enables the server to calculate $\sum \boldsymbol{x}_i$ while preserving the privacy of $\boldsymbol{x}_i$. To achieve this, each user $u_i$ adds a pair-wise additive mask to $\boldsymbol{x}_i$ (assuming a total order of users) to obtain $\boldsymbol{y}_i$, which is then uploaded to the server instead of $\boldsymbol{x}_i$.
$$\boldsymbol{y}_{i}=\boldsymbol{x}_{i}+\sum_{i<j} \text{PRG}(s_{i, j})-\sum_{i>j} \text{PRG}(s_{j, i})$$
where PRG is a pseudorandom generator that generates a sequence of random numbers using the input seed agreed by the user pairs. It is evident from the equation that the masks will be canceled after the aggregation of all masked models, ensuring $z=\sum \boldsymbol{y}_{i} = \sum \boldsymbol{x}_{i}$. Additionally, the seeds are shared among all users using the Shamir secret sharing scheme \cite{shamir1979share} to handle dropout users.

Within the SecAgg scheme, every pair of users engages in a Diffie-Hellman-based key exchange protocol \cite{diffie2022new} for seed agreement. Recent advancements in secure aggregation, such as the SecAgg+ scheme \cite{bell2020secure}, have significantly improved protocol efficiency. This improvement is achieved by enabling each user to communicate with only a subset of users, rather than all users, thereby reducing communication complexity from $\mathcal{O}(N^2)$ to $\mathcal{O}(\log N)$. SecAgg+ accomplishes this by replacing the complete communication graph with a $m$-regular graph, as illustrated in Figure \ref{fig:secagg_p}. In essence, each user communicates with only $m=\mathcal{O}(\log N)$ users to reach a consensus on seeds for generating pair-wise masks, resulting in a substantial reduction in communication costs. Additional input validation for SecAgg is discussed in \cite{bell2023acorn} and the reduction of the per-round setup cost in multi-round SecAgg is explored in \cite{ma2023flamingo}. Therefore, in this paper, we adopt the SecAgg+ scheme \cite{bell2020secure} as the underlying protocol for secure aggregation to construct our proposed scheme.

\begin{figure}[htbp!]
    \centering
    \begin{subfigure}{0.22\textwidth}
        \centering
        \includegraphics[width=1\linewidth]{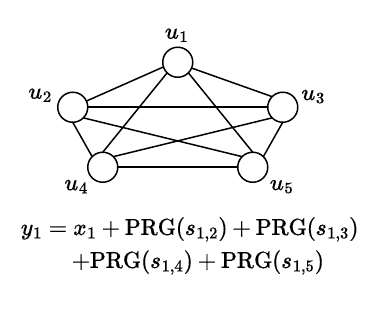}
        \caption{SecAgg \cite{bonawitz2017practical}}
        \label{fig:secagg_graph}
    \end{subfigure}
    \begin{subfigure}{0.22\textwidth}
        \centering
        \includegraphics[width=1\linewidth]{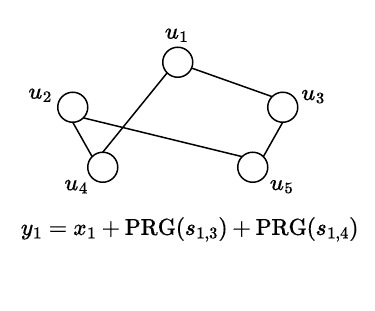}
         \caption{SecAgg+ \cite{bell2020secure}}
        \label{fig:secagg_p_graph}
    \end{subfigure}
    \caption{The sparse communication graph of SecAgg+ \cite{bell2020secure} compared to the complete communication graph of SecAgg \cite{bonawitz2017practical}. We can observe that each user in (b) communicates with much fewer users than those in (a).}
    \label{fig:secagg_p}
\end{figure}

Note that in $t$-out-of-$k$ or ($t$, $k$) Shamir secret sharing schemes \cite{shamir1979share}, a secret is divided into $k$ pieces called shares, satisfying two conditions, including (i) the secret can be reconstructed by any combination of $t$ data pieces, and (ii) the secret cannot be reconstructed by any set of shares with a number less than $t$.

\subsection{Hypergeometric distribution}
\label{sec:hypergeometric_distribution}
The hypergeometric distribution is the probability distribution of a hypergeometric random variable. For a hypergeometrically distributed random variable $X$ with parameters $N, pN, rN$ where $0<p<1$ and $0<r<1$, we use $X\sim\text{HG}(N, pN, rN)$ to denote that the variable $X$ is sampled from a hypergeometric experiment, representing the number of successful draws from $N$ objects with a specific feature containing $pN$ objects with that feature in $qN$ draws without replacement. The tail bounds of the hypergeometric distribution $X\sim\text{HG}(N, pN, rN)$ are given as follows:
\begin{align}
 &   \text{Pr}[X\geq(p+w)rN]\leq e^{-2w^2rN} \nonumber \\
   & \text{Pr}[X\leq(p-w)rN]\leq e^{-2w^2rN} \nonumber 
\end{align}
where $\text{Pr}[\cdot]$ throughout this paper denotes probability.
Note that when $r$ is very small, we can approximate the sampling as the one with replacement to bound the tail of the hypergeometric distribution just like the tail of the binomial distribution \cite{Chvatal79}. The expected value and variance of such hypergeometric distribution are $prN$ and $r(1-r)p(1-p)N^2/(N-1)$, respectively.

\subsection{$m$-regular graph}
$m$-regular graph is a graph where each vertex has $m$ neighbors, i.e., the degree of each vertex is $m$. An $(k,m)$ Harary graph Harary graph is a particular example of $m$-regular graph with $k$ vertices having the smallest possible number of edges. A simple construction of an $(k,m)$ Harary graph as used in SecAgg+ \cite{bell2020secure} is to first put $k$ vertices in a circle, and then connect each vertex to the $\frac{m}{2}$ vertices to its left and the $\frac{m}{2}$ vertices to its right. In this paper, we will also heavily leverage the Harary graph based construction for the adoption of the SecAgg+ protocol.

\subsection{FL model convergence}
\label{sec:convergence}

For a given machine learning problem to be solved by a $k$-user FL, if each user $u_i$'s loss function $L_i$ is $\rho$-smooth, $\mu$-strongly convex (see the definitions in \cite{LiHYWZ20}), the variance of local stochastic gradients of each user is bounded by $\lambda_i^2$, and the expected squared norm of stochastic gradients are uniformly bounded by $G^2$, according to the convergence analysis in \cite{LiHYWZ20}, we can have the following inequality.
\begin{equation}
\label{equ:conv_error}
\begin{aligned}
&\mathbb{E}[L(\boldsymbol{x}^{(j)})]\hspace{-1pt}-\hspace{-1pt}L^{*}\hspace{-1pt}\leq\hspace{-1pt}\frac{\rho}{\theta\hspace{-1pt}+\hspace{-1pt}Ej\hspace{-1pt}-\hspace{-1pt}1}(\frac{2(\alpha\hspace{-1pt}+\hspace{-1pt}\beta)}{\mu^2} \hspace{-1pt} +\hspace{-1pt}\frac{\theta}{2} \mathbb{E}\left\|\boldsymbol{x}^{(0)}\hspace{-1pt}-\hspace{-1pt}\boldsymbol{x}^{*}\right\|^2),
\end{aligned}
\end{equation}
where $\alpha=\sum_{i=1}^k w_i^2\lambda_i^2+6\rho\Gamma+8(E-1)^2G^2$, $\beta=\frac{4(k-\mathcal{K})E^2G^2}{\mathcal{K}(k-1)}$, $\Gamma=L^*-\sum_{i=1}^k w_i L_i^*$ and $\theta=\max\{\frac{8\rho}{\mu},E\}$. Here, $j$ is the number of FL rounds. $\mathcal{K}$ is the number of users who participated in every FL round. $\boldsymbol{x}^{(j)}$ is the global FL model in the $j$-th round. $E$ is the number of local iterations between two FL rounds. $w_i$ is the weight for each user's model for aggregation. $\boldsymbol{x}^*$ and $\boldsymbol{x}_i^*$ are optimal model parameters to minimize the global loss function $L$ and user $u_i$'s loss function $L_i$, respectively. 
Besides, let $M$ denote the number of rounds needed for the convergence of the global FL model. To achieve a fixed precision $\epsilon$, we have
\begin{equation}
\label{equ:conv_rate}
    M = \mathcal{O}[\frac{1}{\epsilon}((1+\frac{1}{\mathcal{K}})EG^2+\frac{\sum_{i=1}^k w_i^2\lambda_i^2+\Gamma+G^2}{E}+G^2)]
\end{equation}

\subsection{Notations}
We provide the parameters used throughout the rest of the paper in Table \ref{tab:my-table}.

\begin{table}[!htbp]
\centering
\caption{Notations of parameters used throughout the paper.}
\label{tab:my-table}
\begin{tabular}{cl}
\hline
\multicolumn{1}{l}{Notation} & Description                                \\ \hline
$N$                           & Number of users                         \\
$k$                           & Cluster size                         \\
$s$                           & Number of clusters                         \\
$\tau$                           & Number of unlearning requests                         \\
$t$                           & Threshold for Shamir secret sharing \\
$q$                      & Maximum number of unlearned users within a cluster       \\ 
$U$                           & Set of users \\
$A$                           & Set of adversarial users \\
$D$                           & Set of dropout users  \\
$Q$                           & Set of unlearned users \\
$C$                           & Set of clusters \\
$\xi$                         & Shamir threshold rate                       \\
$\sigma$                      & Statistical security parameter            \\
$\eta$                        & Correctness parameter                     \\
$\gamma$                      & Maximum fraction of adversarial users     \\
$\delta$                      & Maximum fraction of dropout users       \\
$\zeta$                      & Maximum fraction of unlearned users within a cluster       \\ 
\hline
\end{tabular}
\end{table}

\section{Requirements on Clustering}
\label{sec:requirements_on_clustering}

In this section, we will first describe the threat model and provide an overview of our proposed scheme. We will then present the security requirements for clustering to provide privacy guarantees based on the SecAgg+ protocol. Following this, we will explore user clustering taking into account all the previously mentioned security requirements.

\textbf{Threat model.} Our work considers the following threat model: all FL and FU participants, including users and the central server, can be semi-honest adversaries. In other words, all adversaries try to deduce the private information of honest participants without deviating from the protocol execution. Besides, when corruption happens, some of the adversaries, e.g., a set of users and the central server, may collude to improve their capabilities of learning information, while the other participants remain honest. In addition, although target users can be honest or semi-honest users, we assume the unlearning requests are initiated to unlearn honest users, which represents the worst-case scenario, as removing honest users from clusters poses more security risks compared to removing adversarial users.

\textbf{An overview.} In our proposed scheme, users are grouped into clusters using our proposed clustering algorithm. To illustrate, we specify several stages including (i) During FL training, users in each cluster collaboratively train a global model for their cluster following an FL approach, and the SecAgg+ protocol is adopted within each cluster for privacy-preserving aggregation. (ii) During inference, the test data is fed into each cluster model, and the inferences from all clusters are collected for majority voting. (iii) During unlearning, to remove specific target users, only the clusters containing those target users need to be retrained, as illustrated in Figure \ref{fig:clustered-fu}.


\subsection{Security Requirements on Clustering}
For an FL and FU system with $N$ users, where $u_i \in U$, we denote the set of adversarial users as $A$, dropout users as $D$, and unlearned users as $Q$. The maximum fraction of adversarial users is $\gamma$, and the maximum fraction of dropout users is $\delta$. Note that the assumptions regarding the bounds on the maximum number of adversarial and dropout users are commonly adopted in previous SecAgg works, which can be obtained from FL systems as prior knowledge \cite{bonawitz2017practical,bell2020secure,liu2022efficient,ma2023flamingo,liu2023long}. The exact bound ratio can be adjusted according to the scenario, hence it does not affect the compatibility and scalability of our proposed scheme. Additionally, we divide the $N$ users into $s$ clusters according to our proposed clustering algorithm, with each cluster consisting of $k$ users. Therefore, a $t$-out-of-$k$ Shamir scheme is adopted within each cluster $c_i$ to align with the SecAgg+ protocol.

As mentioned earlier, to enable the adoption of the SecAgg+ protocol, it is crucial to guarantee that the conditions for using Shamir secret sharing schemes are satisfied within each cluster. Therefore, the first requirement $R_1$ in terms of the $t$-out-of-$k$ Shamir secret sharing scheme is that after clustering, the number of adversarial users is less than the threshold $t$. Otherwise, adversarial users can collude to reconstruct the secrets of honest users. Denote the users assigned to the cluster that consists of user $u_i$ as $c(u_i)$, we can define such a requirement $R_1$ as follows.

\begin{requirement}
\label{req:shamir_security}
    (Shamir security) For $|A|\leq \gamma N$, we define the requirement $R_1$  as
    $$R_1(U,A,t)=1~\text{iff}~\forall u_i\in U: |c(u_i)\cap A|< t.$$
\end{requirement}

Additionally, for a $t$-out-of-$k$ Shamir secret sharing scheme, it is essential to ensure that there are sufficient remaining users within each cluster, excluding dropout users and unlearned users. This consideration leads to the second requirement, denoted as $R_2$, of the clustering algorithm, defined as follows.

\begin{requirement}
\label{req:shamir_correctness}
    (Shamir correctness) For $|D \subseteq U|\leq \delta N $ and $Q$, we define the requirement $R_2$  as
    $$R_2(U,D,Q,t)=1~\text{iff}~\forall u_i \in U: |c(u_i)\cap (U\setminus \{D\cup Q\})|\geq t.$$ 
\end{requirement}

Furthermore, to protect the privacy of users' inputs, SecAgg+ adopts pairwise masking technique, which requires each pair of users to agree on a seed to generate masks (see Figure \ref{fig:c1} where the edge represents the two users involved in pairwise masking). The masks can be canceled when computing the sum if each user sends pairwise masked input to the server. In other words, the server can obtain the sum from the set of users that are involved in pairwise masking agreements. However, if some users drop before sending their masked gradients to the server or unlearning requests are initiated to remove some users from clusters, the graph may become disconnected, allowing the server to obtain partial sums instead of the sum from all users, thereby compromising the security (see Figure \ref{fig:c2} as an example where users $u_2,u_4,u_6$ drop, $u_8$ is removed due to unlearning, and the server obtains partial sums from user $u_1,u_5$ and $u_3,u_7$). Therefore, we need to ensure that after some users drop, the graph remains connected. Furthermore, when considering adversarial users, the connectivity of alive honest users should still be maintained (see Figure \ref{fig:c3} as an example where users $u_4,u_6$ drop and $u_8$ is removed, and the server can still obtain partial sums from users $u_1,u_5$ and $u_3,u_7$ by colluding with the adversarial user $u_2$). We define the requirement $R_3$ regarding connectivity security as follows.

\begin{requirement}
\label{req:connectivity}
    (Connectivity security) For $|A|\leq \gamma N$, $|D|\leq \delta N$ and $Q$, we define the requirement $R_3$  as
    \begin{equation*}
        \begin{aligned}
        &R_3(U,A,D,Q)=1 \\
        &\text{iff}~\forall u_i \in U\setminus \{A\cup D \cup Q\}:~c(u_i)~\text{is connected}
        \end{aligned}
    \end{equation*}
\end{requirement}

\begin{figure}[htbp!]
    \centering
    \begin{subfigure}{0.15\textwidth}
        \centering
        \includegraphics[width=0.8\linewidth]{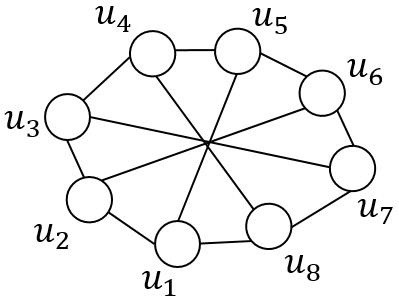}
        \caption{}
        \label{fig:c1}
    \end{subfigure}
    \begin{subfigure}{0.15\textwidth}
        \centering
        \includegraphics[width=0.8\linewidth]{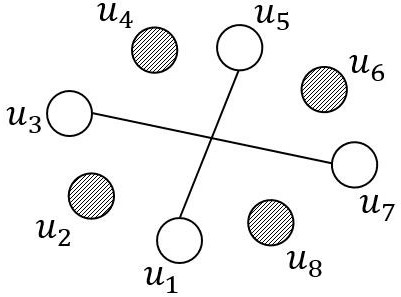}
        \caption{}
        \label{fig:c2}
    \end{subfigure}
    \begin{subfigure}{0.15\textwidth}
        \centering
        \includegraphics[width=0.8\linewidth]{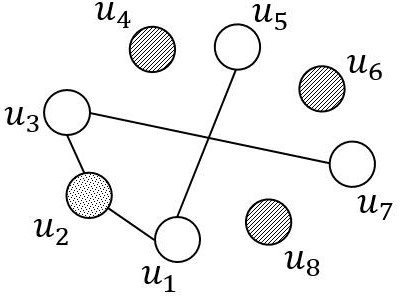}
        \caption{}
        \label{fig:c3}
    \end{subfigure}
    \caption{An illustrative example of graph connectivity where the edge represents that the two users are involved in pairwise masking. In (b), $u_2, u_4, u_6$ are dropout users, $u_8$ is an unlearned user, and in (c) $u_4, u_6$ are dropout users, $u_8$ is an unlearned user, while $u_2$ is an adversarial user.}
    \label{fig:connectivity}
\end{figure}

\subsection{User Clustering}
So far, we have examined the security and correctness requirements associated with the adoption of the SecAgg+ protocol. In the subsequent lemmas, we will demonstrate the process of selecting parameters to meet these requirements.

The first lemma outlines the procedure for selecting the cluster size $k$ to satisfy the Requirement \ref{req:shamir_security} with all but negligible probability regarding Shamir security, given the statistical security parameter $\sigma$ and other system parameters.

\begin{lemma}
\label{lemma:shamir_security}
(Shamir security) For a user $u_i \in U$, its cluster $c(u_i)$ consists of $k$ users. Given a parameter $\xi = t/k$ such that $\xi > \gamma$, if the set of adversarial users $A \subseteq U$ satisfies $|A| \leq \gamma N$ and the cluster size $k$ fulfills the following inequality~(\ref{equ:shamir_security}), then $\text{Pr}[R_1(U, A, t) = 1] \geq 1 - 2^{-\sigma}.$
\begin{equation}
\label{equ:shamir_security}
f_1(k)=2 (\gamma^2+\xi^2-2\xi \gamma) k+(\ln k -\ln N - \sigma \ln 2) \geq 0
\end{equation}
\end{lemma}
\begin{proof}
Let $X_j$ denote the random variable representing the number of adversarial users in cluster $c_j$. In the process of constructing a cluster without replacement through a single draw from a hypergeometric distribution $\text{HG}(N, \gamma N, k)$, where $N$ is the total number of users, $\gamma$ is the maximum fraction of adversarial users, and $k$ is the cluster size, we approximate the sampling as with replacement when $k/N=s^{-1}$ is small \cite{Chvatal79}. This allows us to bound the tail of the hypergeometric distribution, akin to the binomial, as $\text{Pr}[X\geq(\gamma+w)k]\leq e^{-2w^2k}$. By setting $w=(t-\gamma k)/k$, we derive $\text{Pr}[X\geq t]\leq e^{-2(t- \gamma k)^2/k}$, where $t$ is the threshold of the Shamir scheme. For $s=N/k$ draws to construct $s$ clusters, applying Boole's inequality \cite{boole1847mathematical}, we get $1-\text{Pr}[R_1=1] \leq N/k \text{Pr}[X\geq t] \leq N/k e^{-2(t- \gamma k)^2/k}$. Setting $N/k e^{-2(t- \gamma k)^2/k} \leq 2^{-\sigma}$ results in a quadratic-like inequality with a natural logarithm $ak^2+bk+c \geq 0$, where $a=2 \gamma^2, b= (\ln k -\ln N - \sigma \ln 2 -4t\gamma), c=2t^2$ to ensure the Shamir security $\text{Pr}[R_1 = 1] \geq 1 - 2^{-\sigma}$. By substituting $t$ with $\xi k$, Inequality \ref{equ:shamir_security} can be obtained, thereby concluding the proof.
\end{proof}

Next, we turn our attention to the requirement regarding the Shamir correctness. To ensure the correct reconstruction of seeds of dropout users, the number of remaining users, including both honest users and semi-honest adversarial users excluding the target users, should not be less than the threshold $t$ with all but negligible probability. The following lemma outlines the procedure for selecting the cluster size $k$ to satisfy the requirement with all but negligible probability regarding the Shamir correctness.

\begin{lemma}
\label{lemma:shamir_correctness}
(Shamir correctness) For a user $u_i \in U$, its cluster $c(u_i)$ consists of $k$ users. Given the parameter $\xi = t/k$ with $\xi > \delta$ and the parameter $\zeta=q/k$ with $\xi > \zeta$, if the set of dropout users $D \subseteq U$ that $|D|\leq \delta N$, the maximum fraction of unlearned users within a cluster $\zeta$, and the cluster size $k$ fulfills the following inequality~(\ref{equ:shamir_correctness}), then $\text{Pr}[R_2(U,D,Q,t)=1]\geq 1-2^{-\eta}$.
\begin{equation}
\label{equ:shamir_correctness}
\begin{aligned}
f_2(k)&=2 ((1-\delta)^2+\xi^2-2(\xi+\zeta) \delta+2\xi+2\zeta)k\\
& +(\ln k -\ln N - \eta \ln 2) \geq 0
\end{aligned}
\end{equation}
\end{lemma}
\begin{proof}
Let $Y_j$ denote the random variable representing the number of dropout users in cluster $c_j$. In the process of constructing a cluster without replacement through a single draw from a hypergeometric distribution $\text{HG}(N, \delta N, k)$, where $N$ is the total number of users, $\gamma$ is the maximum fraction of dropout users, and $k$ is the cluster size, we approximate the sampling as with replacement when $k/N=s^{-1}$ is small \cite{Chvatal79}. This allows us to bound the tail of the hypergeometric distribution, akin to the binomial, as $\text{Pr}[X\geq(\delta+w)k]\leq e^{-2w^2k}$. 
By setting $w=(\delta k-k-t-q)/k$, we derive $\text{Pr}[X\geq k-t-q]\leq e^{-2(\delta k-k-t-q)^2/k}$, where $t$ is the threshold of the Shamir scheme. For $s=N/k$ draws to construct $s$ clusters, applying Boole's inequality \cite{boole1847mathematical}, we get $1-[R_2=1] \leq N/k \text{Pr}[X\geq k-t-q] \leq N/k e^{-2(\delta k-k-t-q)^2/k}$. Setting $N/k e^{-2(\delta k-k-t-q)^2/k} \leq 2^{-\eta}$ results in a quadratic-like inequality with a natural logarithm $ak^2+bk+c \geq 0$, where $a=2 (1-\delta)^2, b= (\ln k -\ln N - \eta \ln 2 -4\delta(t+q)+4(t+q)), c=2(t+q)^2$ to ensure the Shamir security $\text{Pr}[R_2 = 1] \geq 1 - 2^{-\eta}$. Substituting $t$ with $\xi k$ and $q$ with $\zeta k$ leads to Inequality~\ref{equ:shamir_correctness}, thereby concluding the proof.
\end{proof}

The following lemma outlines the procedure for selecting the cluster size $k$ to meet Requirement $R_3$ for connectivity security of $m$-regular graphs.

\begin{lemma}
\label{lemma:connectivity_security}
(Connectivity security) For any generated $m$-regular graph, if the set of dropout users $D \subseteq U$ that $|D|\leq \delta N$, the set of adversarial users $A \subseteq U$ that $|A|\leq \gamma N$, the maximum fraction of unlearned users within a cluster $\zeta=q/k$, and the cluster size $k$ fulfills the following inequality, then $\text{Pr}[R_3(U,A,D,Q)=1]\geq 1-2^{-\sigma}$.
\begin{equation}
\label{equ:connectivity_security}
f_3(k)=-\ln k \ln(k(\gamma+\delta+\zeta))+2\sigma\ln2 \geq 0
\end{equation}
\end{lemma}
\begin{proof}
Similar to SecAgg+ \cite{bell2020secure}, the security of the third requirement $R_3$ pertaining to connectivity relies on the unique properties of the Harary graph. The Harary graph is a specific type of $m$-connected graph where vertices have the minimum possible number of edges. In the construction of an $(k, m)$ Harary graph, as utilized in SecAgg+, $k$ vertices are initially arranged in a circle, with each vertex then connected to the $\frac{m}{2}$ vertices on its left and the $\frac{m}{2}$ vertices to its right. Disconnecting the Harary graph necessitates the removal of at least $\frac{m}{2}$ successive vertices. To elaborate, to disconnect a user $u_i$, the subsequent $\frac{m}{2}$ users must be in the set $A \cup D \cup Q$. Applying the Boole's inequality \cite{boole1847mathematical} across all $k$ users yields the probability $\text{Pr}[R_3(U, A, D, q) = 0] \leq k(\gamma+\delta+\zeta)^{\frac{m}{2}}$, where $m=\ln{k}$ as same to the setting in SecAgg+. Hence, establishing $\text{Pr} [R_3(U, A, D, q) = 1] \geq 1-2^{-\sigma}$ results in the inequality \ref{equ:connectivity_security}, thereby concluding the proof.
\end{proof}

\subsection{Handling Unlearning Requests}
\label{sec:handling_dynamic_users}

As previously mentioned, dynamic users, such as dropout users who leave the system, can impact the size of clusters, subsequently influencing the privacy guarantees of the SecAgg+ protocol. It has been demonstrated that if the maximum fraction of dropout users is limited by a specific parameter $\delta$, known as prior knowledge, maintaining privacy guarantees with dropout users is feasible as long as the cluster size is sufficiently large. Similar to dropout users, unlearned users can impact the size of clusters, affecting privacy guarantees due to dynamic cluster size. Although a given parameter can limit the numbers of unlearned users within each cluster, the total number of requests for unlearning and new participation is not bounded. This lack of a boundary necessitates an examination of constraints to uphold the security guarantees with such dynamic user types under different scenarios.

Therefore, once the size of each cluster is established by a set of given parameters constrained by the constraints discussed in Section \ref{sec:requirements_on_clustering}, it is crucial to limit the total number of unlearning requests, which ensures that the number of unlearned users for any specific cluster does not exceed the predefined threshold, i.e., the maximum fraction of unlearned users allowed within each cluster $\zeta$. Recall that we are considering a worst-case scenario in which the unlearned users are assumed to be honest. As the system continuously receives unlearning requests, an increasing number of users will be removed from clusters. Therefore, the continual removal of honest users may compromise the guarantees established by Requirements $R_2, R_3$. As such, in this section, we will present a discussion on the requirements for the number of unlearning requests under two settings, including , where unlearning requests arrive one after another and are handled individually, and batch unlearning, where the system processes unlearning through a batch of accumulated sequential requests or from a set of unlearning requests submitted simultaneously, along with the mitigation method to address such risks.

First, we introduce Requirement $R_4$ to ensure that the number of users removed due to unlearning should not exceed a certain limit, preventing a compromise of the security guarantees provided by Requirements $R_2$ and $R_3$ concerning Shamir correctness and connectivity security. The definition of $R_4$ is provided below.

\begin{requirement}
\label{req:capacity}
    (Unlearning capacity) For a set of unlearned users $Q$ with $|Q|=\tau$, the maximum fraction of unlearned users within a cluster $\zeta$, and the cluster size $k$, we define the requirement $R_4$  as
    $$R_4(\zeta,k,Q)=1~\text{iff}~\forall u_i \in U:~ c(u_i)\cap Q \leq \zeta k$$
\end{requirement}

We should note that Requirement $R_4$ leads to different configurations in different unlearning settings. To illustrate, in the sequential unlearning setting, the unlearning requests are processed sequentially. Only after the cluster affected by the unlearning request completes its retraining will the FU system proceed to handle the next unlearning request. In contrast, in the batch unlearning setting, the FU system addresses a batch of unlearning requests simultaneously. If two unlearning requests are initiated to remove two users in the same cluster, retraining for this cluster occurs only once. Next, we will delve into the configuration of these two settings.

\subsubsection{Sequential unlearning}
\label{sec:sequential_unlearning}
In the context of sequential unlearning, it is evident that the initial unlearning request leads to the retraining of $k$ users, consequently removing one user from a cluster. With the second unlearning request, two scenarios arise: (i) the same cluster affected in the first unlearning request is impacted again, leading to the retraining of $k-1$ users with a probability of $1/s$, or (ii) any other cluster is affected, resulting in the retraining of $k$ users with a probability of $(s-1)/s$.

Inductively, when dealing with $\tau_{seq}$ unlearning requests sequentially, each cluster removes a different number of users. Denote the number of removed users in cluster $c_i$ as $q_i$, the probability of the existence of $q_i > q$, where $q = \xi k$ represents the maximum number of unlearned users within a cluster, can be expressed as:
\begin{equation}
\label{equ:qi_existence_sequential}
    Pr[\exists~q_i > q] = 1 - \prod_{i=1}^{s} \sum_{x=0}^{q} \binom{\tau_{seq}}{x} \left(\frac{1}{s}\right)^x \left(1 - \frac{1}{s}\right)^{\tau_{seq} - x},
\end{equation}
where $\binom{a}{b}$ represents the binomial coefficient, which calculates the number of ways to choose $a$ items from a set of $b$ items. 

To meet Requirement $R_4$ with all but negligible probability, it is essential to ensure that $Pr[\exists~q_i > q] \leq 2^{-\sigma}$. Therefore, the following lemma outlines the procedure for constraining the number of unlearning requests, denoted as $\tau_{seq}$, within a given clustering to fulfill the requirement $R_4$ regarding unlearning capacity.

\newcounter{sublemma}
\renewcommand{\thelemma}{\arabic{lemma}.\arabic{sublemma}}
\setcounter{sublemma}{1}
\begin{lemma}
\label{lemma:unlearning_capacity_sequential}
(Unlearning capacity with sequential requests) 
For a sequence of $\tau_{seq}$ unlearning requests, if the maximum fraction of unlearned users within a cluster satisfies $\zeta = q/k < \sqrt{\frac{N^2 \sigma \ln2}{2k^4}}$, and the number of sequential unlearning requests $\tau_{seq}$ fulfills the following inequality~(\ref{equ:unlearning_capacity_sequential}), then the probability $\text{Pr}[R_4(\zeta,k,Q)=1]\geq 1-2^{-\sigma}$.
\begin{equation}
\label{equ:unlearning_capacity_sequential}
\tau_{seq} \leq \sqrt{\frac{N^3\sigma \ln 2}{2k^3}}
\end{equation}
\end{lemma}
\begin{proof}
Given the expression for $Pr[\exists~q_i > q]$ as described in Equation (\ref{equ:qi_existence_sequential}), we can utilize McDiarmid's inequality \cite{doob1940regularity} on each term within the summation over $x$ to establish an upper bound for $\sum_{x=0}^{q} \binom{\tau_{seq}}{x} (\frac{1}{s})^x (1 - \frac{1}{s})^{\tau_{seq} - x}$, yielding $\sum_{x=0}^{q} \binom{\tau_{seq}}{x} (\frac{1}{s})^x (1 - \frac{1}{s})^{\tau_{seq} - x} \leq e^{-2q^2/s^2}$. By substituting this upper bound back into the original expression in Equation (\ref{equ:qi_existence_sequential}) and setting the security condition $Pr[\exists~q_i > q] \leq 2^{-\sigma}$, we can derive the inequality given in Equation (\ref{equ:unlearning_capacity_sequential}) when $\zeta < \sqrt{\frac{N^2 \sigma \ln 2}{2k^4}}$, thereby concluding the proof.
\end{proof}

\subsubsection{Batch unlearning}
\label{sec:batch_unlearning}
In contrast to the setting of sequential unlearning, the FU system handles a batch of unlearning requests concurrently. When two unlearning requests are initiated to remove two users within the same cluster, the retraining process for that cluster takes place only once. Similarly, denote the number of removed users in cluster $c_i$ as $q_i$, the probability of the existence of $q_i > q$, where $q = \xi k$ represents the maximum number of unlearned users within a cluster, can be expressed as:
\begin{equation}
\label{equ:qi_existence_batch}
    Pr[\exists~q_i \hspace{-1pt} > \hspace{-1pt} q]  \hspace{-1pt}= \hspace{-1pt} (\binom{\tau_{bat}  \hspace{-1pt}+ \hspace{-1pt} s  \hspace{-1pt}- \hspace{-1pt} 1}{s - 1} - \sum_{i=1}^{s-q} \binom{\tau_{bat}  \hspace{-1pt}+  \hspace{-1pt}s  \hspace{-1pt}- \hspace{-1pt} 1  \hspace{-1pt}- \hspace{-1pt} i}{s - 1 - i})/s^\tau_{bat}.
\end{equation}
To meet Requirement $R_4$ with all but negligible probability, it is essential to ensure that $Pr[\exists~q_i > q] \leq 2^{-\sigma}$. Therefore, the following lemma outlines the procedure for constraining the number of unlearning requests $\tau_{bat}$, within a given clustering to fulfill the requirement $R_4$ regarding unlearning capacity.

\addtocounter{lemma}{-1}
\addtocounter{sublemma}{1}
\begin{lemma}
\label{lemma:unlearning_capacity_batch}
(Unlearning capacity with batch requests) For a batch of $\tau_{bat}$ unlearning requests, if the cluster size $k$ satisfies $N^2/k^2 -2\sigma \ln 2 + 2\ln(N/k) > 0$, and the maximum fraction of unlearned users within a cluster $\zeta = q/k$, the number of unlearning requests $\tau_{bat}$ in the batch fulfills the following inequality, then the probability $\text{Pr}[R_4(\zeta,k,Q)=1]\geq 1-2^{-\sigma}$.
\begin{equation}
\label{equ:unlearning_capacity_batch}
f_1(\tau_{bat})=\frac{2(N\zeta-\tau_{bat} k)^2}{k\tau_{bat}(N-k)}+\ln k-\ln N-\sigma\ln 2 \geq 0
\end{equation}
\end{lemma}
\begin{proof}
Recall that $q_i$ represents the number of unlearned users in the cluster $c_i$, with each unlearning request having a probability of affecting the cluster $c_i$ equal to $1/s$. Therefore, we can obtain the Chernoff bound \cite{chernoff1952measure} for each $q_i$ as $e^{-\frac{2s^2(q-\tau_{bat}/s)^2}{\tau_{bat}(s-1)}}$. By applying Boole's inequality \cite{boole1847mathematical}, we can get the union bound as $Pr[\exists~q_i > q] \leq se^{-\frac{2s^2(q-\tau_{bat}/s)^2}{\tau_{bat}(s-1)}}$. Setting the security condition $se^{-\frac{2s^2(q-\tau_{bat}/s)^2}{\tau_{bat}(s-1)}} \leq 2^{-\sigma}$, we can derive the inequality given in Equation (\ref{equ:unlearning_capacity_batch}), thereby concluding the proof.
\end{proof}

\renewcommand{\thelemma}{\arabic{lemma}}

\section{Putting it all together}

So far, we have established requirements for the proposed clustered FU scheme, encompassing (i) Shamir security, (ii) Shamir correctness, (iii) connectivity security, and (iv) unlearning capacity. Now, given the statistical security parameter $\sigma$ and the correctness parameter $\eta$ along with other parameters, we can outline the combined requirements for a $(\sigma,\eta)$-good clustered FU scheme, as presented in the theorem below.

\setcounter{definition}{0}
\begin{definition}
\label{thm:good_graph}
(Good clustered FU) Given parameters $N, \gamma, \delta, \xi, \zeta, \sigma, \eta$, where $\xi = t/k$, $\zeta = q/k$, and $\gamma + \delta < 1$, a clustered FU system, which divides $N$ clients into $s$ clusters, each containing $k$ clients and handles $\tau$ unlearning requests, is considered $(\sigma, \eta)$-good if it satisfies the following inequalities.
\begin{enumerate}
    \item $Pr(R_{1}=1 \wedge R_{3}=1 \wedge R_{4}=1) \geq 1-2^{-\sigma}$
    \item $Pr(R_{2}=1) \geq 1-2^{-\eta}$
\end{enumerate}
\end{definition}

Now, to establish a $(\sigma,\eta)$-good clustering-based FU, we introduce our scheme, which is elaborated in Algorithm \ref{alg:fedbr} and Algorithm \ref{alg:par_gen}. Then, based on the previously discussed Lemmas \ref{lemma:shamir_security}, \ref{lemma:shamir_correctness}, \ref{lemma:connectivity_security}, \ref{lemma:unlearning_capacity_sequential}, and \ref{lemma:unlearning_capacity_batch}, we can formulate Theorem \ref{thm:good_clustered_fu} for $(\sigma, \eta)$-good clustered FU schemes.

\begin{algorithm}
\SetAlgoNoEnd
\caption{Our proposed scheme}
\label{alg:fedbr}
\KwIn{A set of $N$ clients, maximum fraction of malicious clients $\gamma$, maximum fraction of dropout clients $\delta$, maximum fraction of unlearned clients within a cluster $\zeta$, Shamir threshold rate $\xi$, statistical security parameter $\sigma$, correctness parameter $\eta$.}
$(k,s,\tau_{seq},\tau_{bat})=\text{ParGen}(N,\gamma,\delta,\zeta,\xi,\sigma,\eta)$\;\tcp{Parameter generation}
Sample a random permutation $\pi:U \mapsto C$ where $C=\{c_i\}^s$ with $|c_i|=k$; \tcp{Clustering}
Randomly assign the rest of $N-sk$ clients to $s$ clusters; \tcp{Clustering the reminder}
\SetKwFunction{FL}{\textbf{Sequential unlearning}}
\SetKwProg{Fn}{}{:}{\KwRet}
\Fn{\FL}{
\textit{Ensure that the number of sequential requests does not reach the unlearning capacity $\tau_{seq}$}\;
Upon receiving a sequential unlearning request, locate the impacted cluster denoted as $c_j$\;
Retrain the cluster $c_j$ and keep FL training on other clusters $\{c_i|i\neq j\}$\;
}
\SetKwFunction{FL}{\textbf{Batch unlearning}}
\SetKwProg{Fn}{}{:}{\KwRet}
\Fn{\FL}{
\textit{Ensure that the number of batch requests does not reach the unlearning capacity $\tau_{bat}$}\;
Upon receiving the batch requests for unlearning, locate the impacted clusters denoted as $C^{\prime}$\;
Retrain the cluster $c_j\in C^{\prime}$ and keep FL training on other clusters $\{c_i|c_i \notin C^{\prime}\}$\;
}
\end{algorithm}

\begin{algorithm}
\SetAlgoNoEnd
\caption{ParGen}
\label{alg:par_gen}
\KwIn{Number of clients $N$, maximum fraction of malicious clients $\gamma$, maximum fraction of dropout clients $\delta$, maximum fraction of unlearned clients within a cluster $\zeta$, Shamir threshold rate $\xi$, statistical security parameter $\sigma$, correctness parameter $\eta$.}
\KwOut{The cluster size $k$, the number of clusters $s$, the unlearning capacity $\tau_{seq}$ for sequential requests and $\tau_{bat}$ for batch requests.}
$k_1= \text{min} \{k|k>0,f_1(k)\geq 0\}$, where $f_1(k)$ is defined in Inequality (\ref{equ:shamir_security}); \\
$k_2= \text{min} \{k|k>0,f_2(k)\geq 0\}$, where $f_2(k)$ is defined in Inequality (\ref{equ:shamir_correctness});\\
$k_3= \text{min} \{k|k>0,f_3(k)\geq 0\}$, where $f_3(k)$ is defined in Inequality (\ref{equ:connectivity_security});\\
$k=\text{max}(\lceil k_1 \rceil,\lceil k_2 \rceil,\lceil k_3 \rceil)$ and $s= \lfloor N/k \rfloor$\;\tcp{}
$\tau_{seq}=\lfloor \sqrt{\frac{N^3\sigma \log 2}{2k^3}} \rfloor$\;\\
$\tau_{bat} =\lfloor \frac{N^2\zeta^2}{N^2/k^2 -2\sigma \log 2 + 2\log(N/k)}\rfloor$\;\\
\KwRet{$(k,s,\tau_{seq},\tau_{bat})$}
\end{algorithm}

\setcounter{theorem}{0}
\begin{theorem}
\label{thm:good_clustered_fu}
(Good clustered FU) For a given set of parameters $N, \gamma, \delta, \zeta, \xi, \sigma, \eta$, where $\xi = t/k, \zeta = q/k$, and $\gamma + \delta < 1$, the clustered FU generated by Algorithm \ref{alg:fedbr} is $(\sigma, \eta)$-good.
\end{theorem}
\begin{proof}
It is evident that $k_1, k_2, k_3$, as deduced in Lines 1, 2, and 3 of Algorithm \ref{alg:par_gen}, represent the minimum values that fulfill the respective requirements of Shamir security for $R_1$ as per Lemma \ref{lemma:shamir_security}, Shamir correctness for $R_2$ as per Lemma \ref{lemma:shamir_correctness}, and connectivity security for $R_3$ as per Lemma \ref{lemma:connectivity_security}, respectively. 
The floor operation applied to $k_1, k_2, k_3$ in line 4 does not compromise the security and correctness guarantees but adds an additional layer of strength. 
This aligns with the deduction of $\tau_{seq}$ and $\tau_{bat}$ as performed in lines 5 and 6, as stipulated by Lemma \ref{lemma:unlearning_capacity_sequential} and Lemma \ref{lemma:unlearning_capacity_batch}, in order to meet the unlearning capacity requirements outlined in $R_4$. Furthermore, to handle the remaining clients, we assign them randomly into clusters, as outlined in lines 2 and 3 of Algorithm \ref{alg:fedbr}. This results in the actual cluster size being slightly larger than the initially designed one. Importantly, this adjustment does not jeopardize the security and correctness guarantees, as per the deductions outlined in the aforementioned Lemmas. As long as we ensure that the number of requests does not exceed the unlearning capacity $\tau_{seq}$ for the sequential setting and $\tau_{bat}$ for the batch setting, the proposed scheme achieves $(\sigma, \eta)$-goodness, thus concluding the proof.
\end{proof}

Note that ParGen($\cdot$) presented in Algorithm \ref{alg:par_gen} is designed to generate the smallest $k$ and the largest $\tau_{seq}$ and $\tau_{batch}$ for a $(\sigma,\eta)$-good clustering-based scheme, in accordance with Theorem \ref{thm:good_clustered_fu}. As an example, consider a situation where $N=200$, $\gamma=\delta=0.1$, $\xi=0.7$, $\zeta=0.1$, $\sigma=40$, and $\eta=40$, the smallest value of the cluster size $k$ required to achieve a $(\sigma,\eta)$-good clustering-based as described in Algorithm \ref{alg:fedbr} is $k=60$. This demonstrates an empirical result showcasing the provision of privacy guarantees for a clustering-based scheme. Furthermore, we offer the following remarks.

\begin{remark}
    (Number of retraining users) First, we provide an analysis of the number of users required to undergo retraining due to unlearning. In a sequential setting, our proposed scheme initiates the retraining of a cluster upon receiving an unlearning request and proceeds to handle the next request once the retraining is completed. Hence, with a total of $\tau_{seq}$ unlearning requests, the expected number of users to undergo retraining is denoted as $\mathbb{E}_{seq}$ and can be simplified using the binomial theorem to $\mathbb{E}_{seq}=(2N+1-\tau_{seq}^2)k/(2N)-\tau_{seq}$. We can establish an upper bound by assuming that after each unlearning request, the size of each cluster remains constant, which leads to $\mathbb{E}_{seq} \leq k\tau_{seq}$. In a batch setting, our proposed triggers the retraining of clusters upon receiving a batch of unlearning requests. Therefore, each cluster undergoes retraining at most once. Since whether a cluster is affected or not can be considered as a Bernoulli random variable, we can easily compute the expected number of users to undergo retraining in this batch setting as $\mathbb{E}_{bat}=N-N(1-k/N)^{\tau_{bat}}-\tau_{bat}$. This leads to asymptotic results where $\mathbb{E}_{bat}$ approaches to  0 when $\tau_{bat}\rightarrow 0$ and $N-\tau_{bat}$ when $\tau_{bat}\rightarrow +\infty$.

\end{remark}

\begin{remark}
    (Choice of cluster size) The cluster size $k$ has a direct impact on the number of users requiring retraining. For a cluster $c_i$, a larger cluster size $k=|c_i|$ increases the likelihood that an unlearning request will impact cluster $c_i$, necessitating a greater number of users to retrain their models. Additionally, larger clusters can slow down the SecAgg+ protocol due to underlying user-pair-wise key exchanges. Conversely, smaller cluster sizes may enhance efficiency by reducing the number of users needing retraining but could result in weaker learners, a topic discussed in both \cite{kearns1988thoughts} and \cite{bourtoule2021machine}. Moreover, cluster size influences the convergence performance of both FL and FU. Hence, selecting the optimal cluster size $k$ is a non-trivial task. We will delve deeper into this aspect by conducting empirical investigations in the experimental evaluation section.
\end{remark}

\section{Convergence Analysis}
\label{sec:convergence_analysis}

As discussed previously, the cluster size $k$ has an impact on the convergence performance of both FL and FU. In this section, we will provide a theoretical analysis of this effect, following the principles outlined in Section \ref{sec:preliminaries}. Hence, we first introduce a metric to quantify the average number of users within a single cluster, as defined below. Note that since unlearning conducted in our proposed scheme is based on retraining, its convergence follows the same principles as those provided for FL. 

\begin{definition}
\label{def:acc}
Average Cluster Cardinality ($\mathcal{H}$). The average cluster cardinality quantifies the average number of users participating in unlearning within a single cluster.
\begin{equation}
\mathcal{H} = {\liminf_{j\to \infty}}\frac{1}{j}\sum_{c_j \in C} \mathbb{E} |c_j|
\end{equation}
\end{definition}

Since our proposed scheme has distinct unlearning capacities for sequential requests and batch requests, the average cluster cardinality also varies between these two settings. Following the analysis of the number of retraining users $\mathbb{E}_{seq}$ and $\mathbb{E}_{bat}$ presented earlier, we can have 
\begin{equation}
\label{equ:acc_seq}
\begin{aligned}
    \mathcal{H}_{seq}&=N(N-\mathbb{E}_{seq})/k \\
    &=(1-\tau_{seq}/N)k-(2N+1-\tau_{seq}^2)/(2N^2)
\end{aligned}
\end{equation}
and
\begin{equation}
\label{equ:acc_bat}
\begin{aligned}
    \mathcal{H}_{bat} \hspace{-1pt}= \hspace{-1pt}N(N-\mathbb{E}_{bat})/k  \hspace{-1pt}= \hspace{-1pt}k+N(N+\tau_{bat})/k-2N.
\end{aligned}
\end{equation}

 Building upon the deductions outlined above, we present the following remarks.

\begin{remark}
    (Model difference) By substituting $\mathcal{K}$ in Equation (\ref{equ:conv_error}) with $\mathcal{H}_{seq}$ from Equation (\ref{equ:acc_seq}) and $\mathcal{H}_{bat}$ from Equation (\ref{equ:acc_bat}), we can derive
\begin{equation}
\label{equ:conv_error_seq}
\begin{aligned}
&\mathbb{E}[L(\boldsymbol{x}^{(j)})]_{seq}-L^{*}\leq\frac{\rho}{\theta+Ej-1}(\frac{2(\alpha+\beta_{seq})}{\mu^2} \\
&+\frac{\theta}{2} \mathbb{E}\left\|\boldsymbol{x}^{(0)}-\boldsymbol{x}^{*}\right\|^2)
\end{aligned}
\end{equation}
and
\begin{equation}
\label{equ:conv_error_bat}
\begin{aligned}
&\mathbb{E}[L(\boldsymbol{x}^{(j)})]_{bat}-L^{*}\leq\frac{\rho}{\theta+Ej-1}(\frac{2(\alpha+\beta_{bat})}{\mu^2} \\
&+\frac{\theta}{2} \mathbb{E}\left\|\boldsymbol{x}^{(0)}-\boldsymbol{x}^{*}\right\|^2)
\end{aligned}
\end{equation}
where $\alpha,\Gamma,\theta$ are defined therein in Equation (\ref{equ:conv_error}) with $\beta_{seq}=\frac{4E^2G^2(2N^2k - 2Nk(N - \tau_{seq}) + 2N - \tau_{seq}^2 + 1)}{(k - 1)(2Nk(N - \tau_{seq}) - 2N + \tau_{seq}^2 - 1)}$ and $\beta_{bat}=\frac{4E^2G^2N(-N + 2k - \tau_{seq})}{(k - 1)(N(N + \tau_{seq}) + k(-2N + k))}$.
    
From Equations (\ref{equ:conv_error_seq}) and (\ref{equ:conv_error_bat}), we can observe that $\beta$ and consequently the upper bound of $\mathbb{E}[L(\boldsymbol{x}^{(j)})]{seq}-L^{*}$ generally tend to increase with $k$ when either $N$ and $\tau_{seq}$ or $\tau_{bat}$ are large, or both are large. Only when both $N$ and $\tau_{seq}$ are small does $\beta$ and $\mathbb{E}[L(\boldsymbol{x}^{(j)})]_{seq}-L^{*}$ tend to decrease as $k$ increases. Since the cluster size $k$ must be greater than a value constrained by Requirements $R_1, R_2, R_3$, it can be proven that the upper bound of the model difference, $\mathbb{E}[L(\boldsymbol{x}^{(j)})]_{seq}-L^{*}$, between the global model converged within a cluster in our proposed scheme and the optimal solution consistently increases with a larger cluster size $k$.
\end{remark}

\begin{remark}
    (Convergence rate) Likewise, by integrating Equation (\ref{equ:conv_rate}) with Equation (\ref{equ:acc_seq}) and Equation (\ref{equ:acc_bat}), we can obtain
\begin{equation}
\label{equ:conv_rate_seq}
    M_{seq} = \mathcal{O}[\frac{1}{\epsilon}((1+\frac{1}{\mathcal{H}_{seq}})EG^2+\frac{\sum_{i=1}^k w_i^2\lambda_i^2+\Gamma+G^2}{E}+G^2)]
\end{equation}
and
\begin{equation}
\label{equ:conv_rate_bat}
    M_{bat} = \mathcal{O}[\frac{1}{\epsilon}((1+\frac{1}{\mathcal{H}_{bat}})EG^2+\frac{\sum_{i=1}^k w_i^2\lambda_i^2+\Gamma+G^2}{E}+G^2)]
\end{equation}

where $\mathcal{H}_{seq}$ and $\mathcal{H}_{bat}$ are derived in Equation (\ref{equ:acc_seq}) and Equation (\ref{equ:acc_bat}), respectively. Similar observations can be drawn from Equation (\ref{equ:conv_rate_seq}) and Equation (\ref{equ:conv_rate_bat}), which align with the earlier discussion. A larger cluster size leads to faster convergence in both federated learning and unlearning but also results in a greater number of users requiring retraining, introducing a trade-off.
\end{remark}

\begin{remark}
    (Multiple clusters) Note that all the convergence analyses provided above pertain to a single cluster. The situation becomes more intricate when dealing with clustering-based involving multiple clusters. As mentioned previously, a smaller cluster size may lead to weaker learners, as discussed in \cite{kearns1988thoughts}. Essentially, the accuracy of a smaller cluster tends to be lower than that of a larger cluster trained on a more extensive dataset. However, a portion of this lost accuracy can be recovered through the voting operation in clustering-based, akin to ensemble learning systems \cite{opitz1999popular}. Therefore, we will conduct empirical investigations in the experimental evaluation section to further explore this phenomenon.
\end{remark}

\section{Experimental Evaluation}
\label{sec:experiments}

\begin{figure*}[htbp!]
    \centering
    \begin{subfigure}{0.3\textwidth}
        \includegraphics[width=1.1\linewidth]{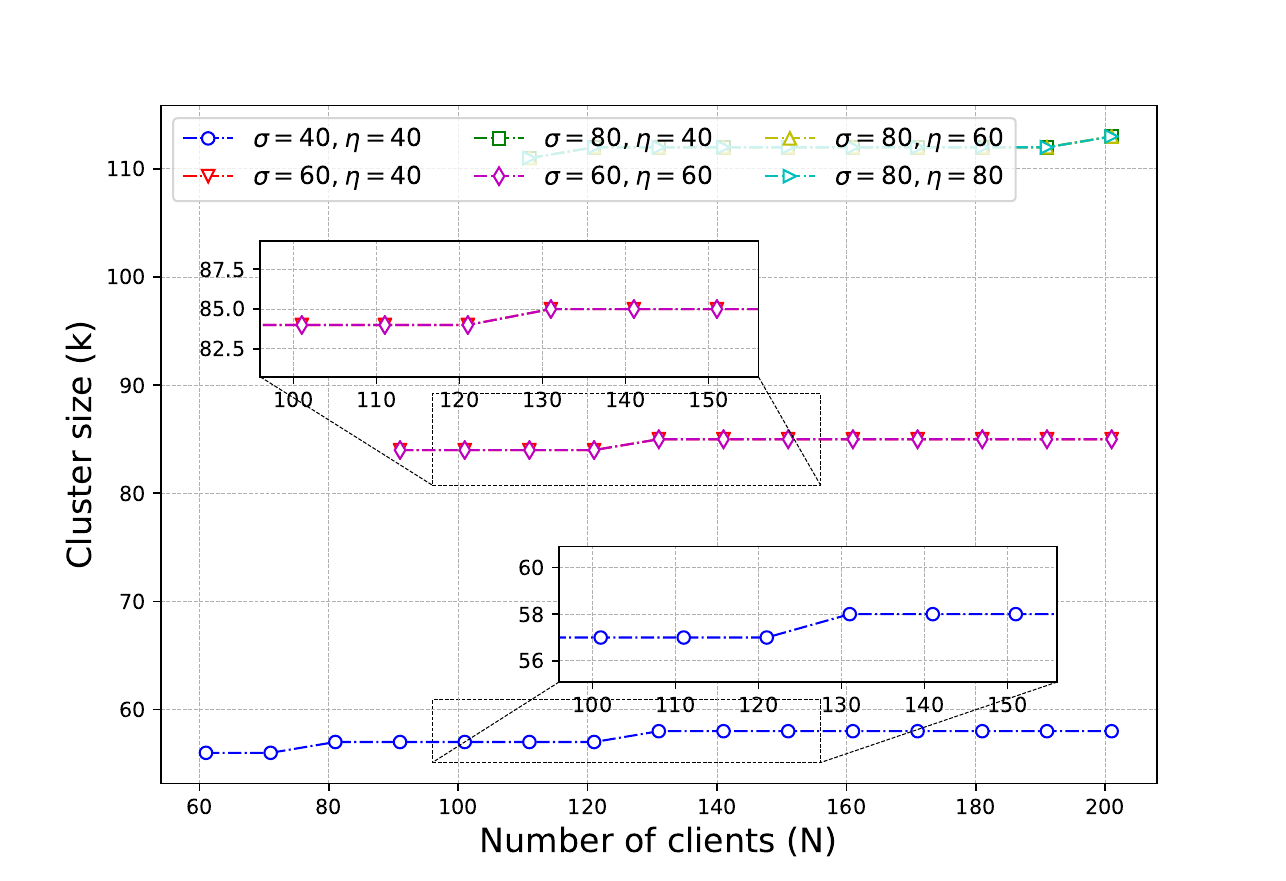}
        \caption{$\sigma\&\eta$}
        \label{fig:cluster_size_sigma_eta}
    \end{subfigure}
    \begin{subfigure}{0.3\textwidth}
        \includegraphics[width=1.1\linewidth]{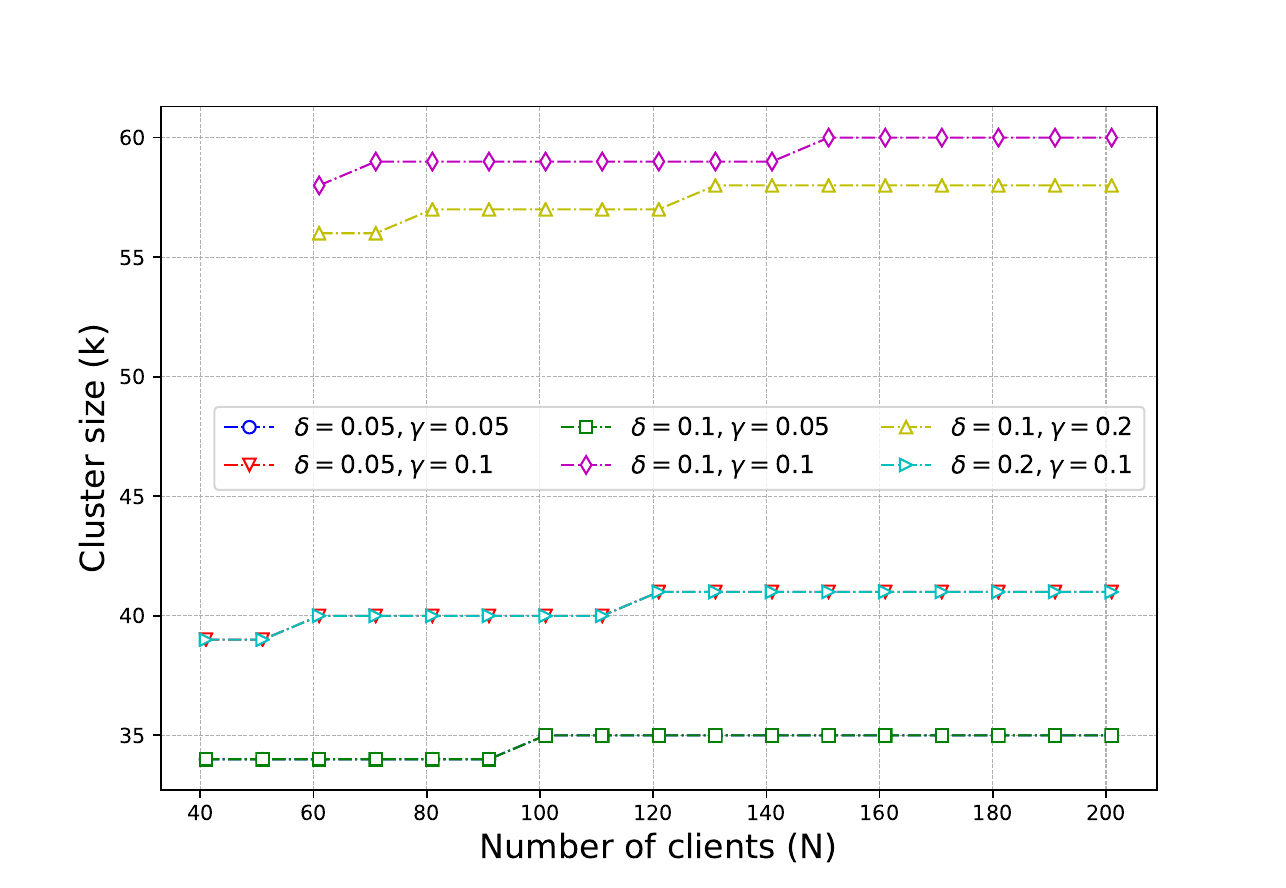}
        \caption{$\delta\&\gamma$}
        \label{fig:cluster_size_delta_gamma}
    \end{subfigure}
    \begin{subfigure}{0.3\textwidth}
        \includegraphics[width=1.1\linewidth]{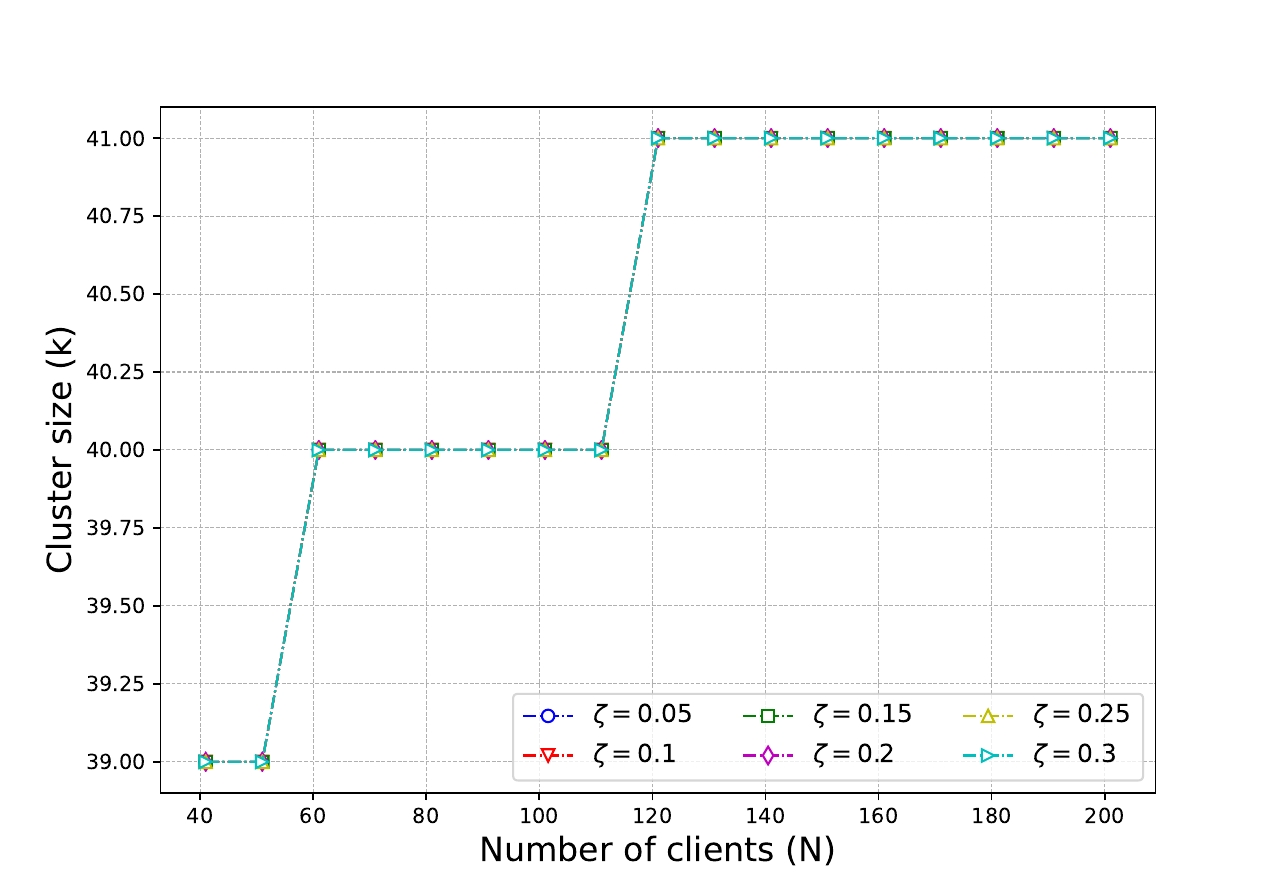}
        \caption{$\zeta$}
        \label{fig:cluster_size_zeta}
    \end{subfigure}
    \caption{Comparison of the required cluster size $k$ across different parameter settings, considering (a) security and correctness parameters $\sigma$ and $\eta$ where the parameters $\{\gamma, \delta,\zeta, \xi\}$ are set to be $\{0.2,0.2,0.1,0.7\}$, respectively; (b) the fraction of adversarial and dropout users $\gamma$ and $\delta$ where the parameters $\{\sigma, \eta,\zeta, \xi\}$ are set to be $\{40,40,0.1,0.7\}$, respectively; and (c) the fraction of unlearned users $\zeta$ within a cluster where the parameters $\{\sigma, \eta, \gamma, \delta,\zeta, \xi\}$ are set to be $\{40, 40, 0.1,0.1,0.1,0.7\}$, respectively.}
    \label{fig:k_results}
\end{figure*}

Our experiments consist of three parts: (i) Firstly, we numerically demonstrate the impact of various parameters on the cluster size and capacity to handle dynamic users. (ii) Next, we evaluate the convergence performance of federated learning and unlearning within our proposed scheme under different parameter settings. (iii) Additionally, we assess the overhead of the underlying SecAgg+ protocol adopted within clusters for privacy-preserving aggregation, which significantly affects the runtime during each federated learning or unlearning round.

\textbf{Experiment settings}
Our experiments are conducted on two image classification datasets, including MNIST and CIFAR-10. For MNIST, we use a set of 60,000 images for training and 10,000 for testing with a simple two-layer CNN denoted as $\mathcal{M}_1$. For CIFAR-10, we utilize ResNet-18 denoted as $\mathcal{M}_2$ using a training set of 50,000 images and a test set of 10,000 images. Both datasets are uniformly distributed among users, with all experiments conducted on a machine equipped with two AMD EPYC 7763 64-Core processors @ 3.5GHz, 2 TB of RAM, a single NVIDIA A100-SXM4-40GB GPU, and with 0.08 ms latency and 34.0 Gbps bandwidth on average. The SecAgg+ protocols are implemented utilizing the open-source code available in \cite{liu2022efficient}.

\textbf{Clustering results.} From Figure \ref{fig:k_results}, we can observe the effect of different parameters on the smallest cluster size $k$ required to satisfy the requirements $R_1,R_2,R_3$ in order to achieve a $(\sigma,\eta)$-good clustering-based.  This demonstrates how varying parameter settings influence the determination of the optimal cluster size that meets the specified criteria within our proposed FU scheme. Specifically, we note that larger values of $\sigma,\eta,\gamma,\delta$ generally lead to a larger required cluster size $k$, albeit with points of convergence. Figure \ref{fig:cluster_size_sigma_eta} illustrates that the security parameter $\sigma$ has a more pronounced impact than the correctness parameter $\eta$. Further insights from Figure \ref{fig:cluster_size_delta_gamma} show that the required cluster size $k$ scales with both $\gamma$ and $\delta$, with $\gamma$ exerting a more significant influence than $\delta$, suggesting that our proposed scheme is more sensitive to adversarial users than to dropout users. Moreover, Figure \ref{fig:cluster_size_zeta} indicates that the maximum fraction of unlearned users within a cluster, $\zeta$, has a negligible impact on determining $k$, highlighting the predominant role of $\sigma,\eta,\gamma,\delta$ in the parameter setting of our proposed scheme. Besides, in Figures \ref{fig:cluster_size_sigma_eta}, \ref{fig:cluster_size_delta_gamma}, and \ref{fig:cluster_size_zeta}, we observe a sudden and substantial shift in the cluster size $k$, which then stabilizes for a period, in response to the increase in the number of users $N$. Figure \ref{fig:tau_results} illustrates the impact of $\delta$ and $\gamma$ on the unlearning capacity to meet Requirement $R_4$, showcasing $\tau_{seq}$ for sequential unlearning requests in Figure \ref{fig:tau_seq}, and $\tau_{bat}$ for batch unlearning requests in Figure \ref{fig:tau_bat}. We observe a consistent increase in the unlearning capacity with the rise in the number of users of our proposed scheme.

\begin{figure}[htbp!]
    \centering
    \begin{subfigure}{0.23\textwidth}
        \includegraphics[width=1.05\linewidth]{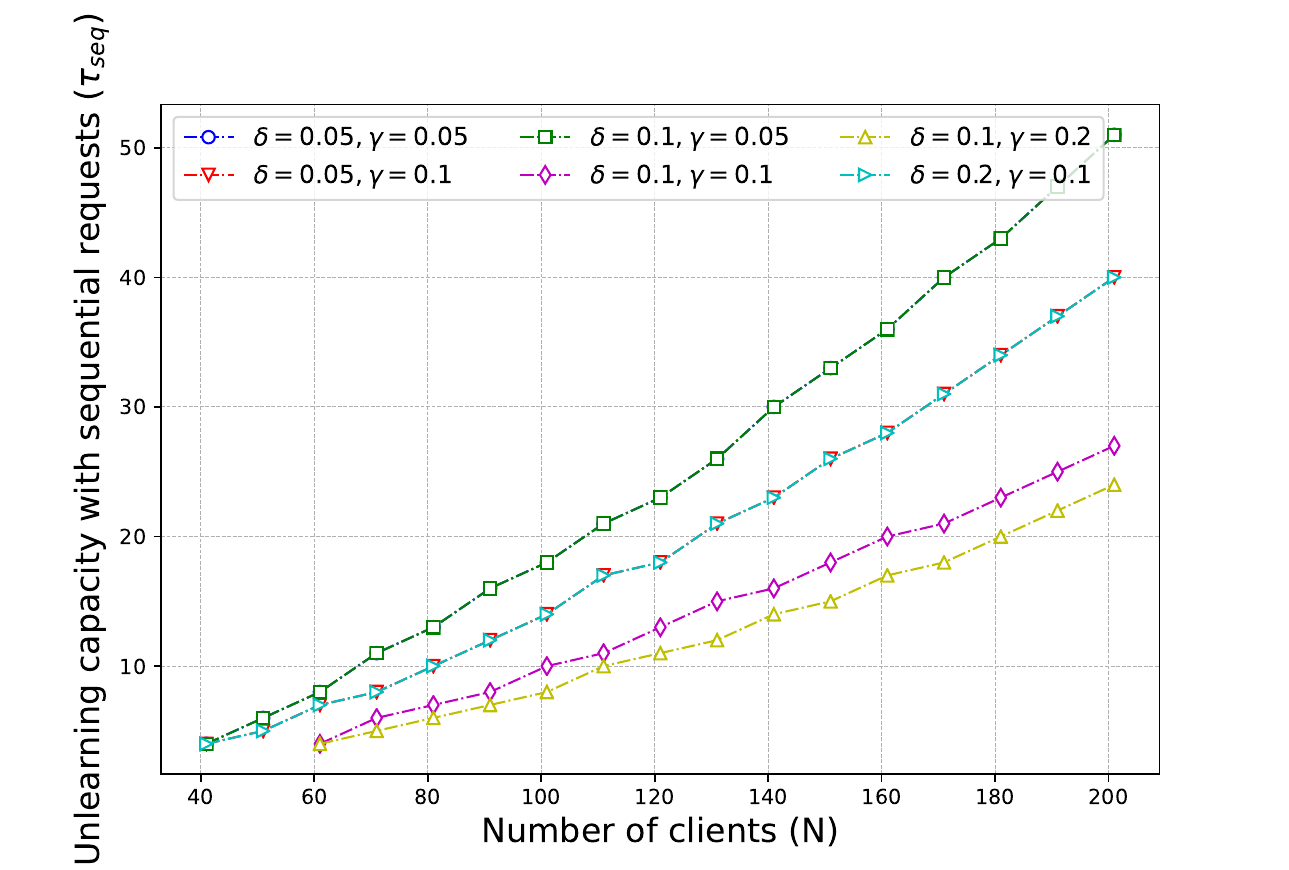}
        \caption{Sequential unlearning}
        \label{fig:tau_seq}
    \end{subfigure}
    \begin{subfigure}{0.23\textwidth}
        \includegraphics[width=1.1\linewidth]{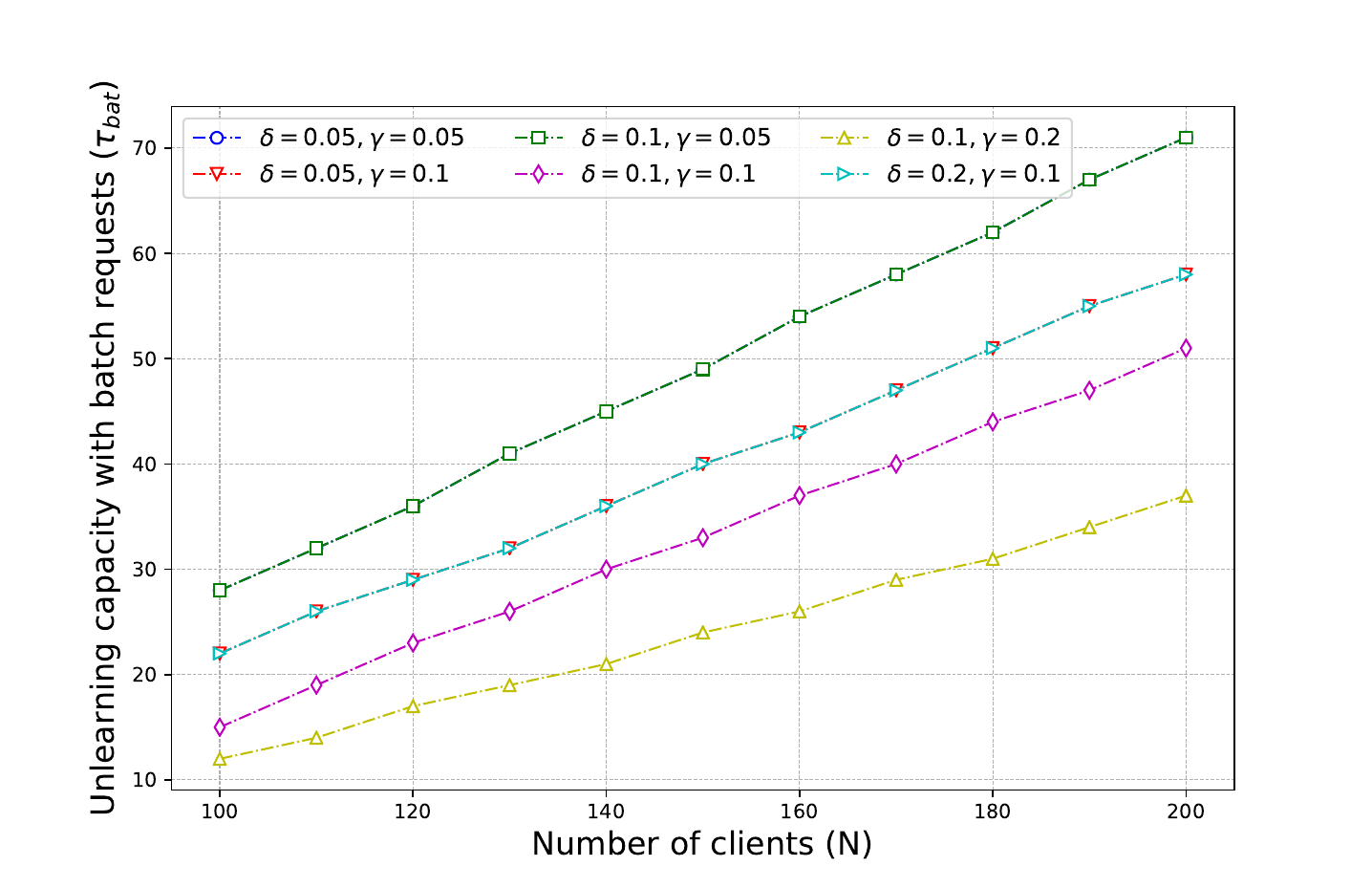}
        \caption{Batch unlearning}
        \label{fig:tau_bat}
    \end{subfigure}
    \caption{Comparison of the unlearning capacities $\tau_{seq}$ for sequential unlearning and $\tau_{bat}$ for batch unlearning under various parameter settings.}
    \label{fig:tau_results}
\end{figure}

\begin{figure}[htbp!]
    \centering
    \begin{subfigure}{0.23\textwidth}
        \includegraphics[width=1\linewidth]{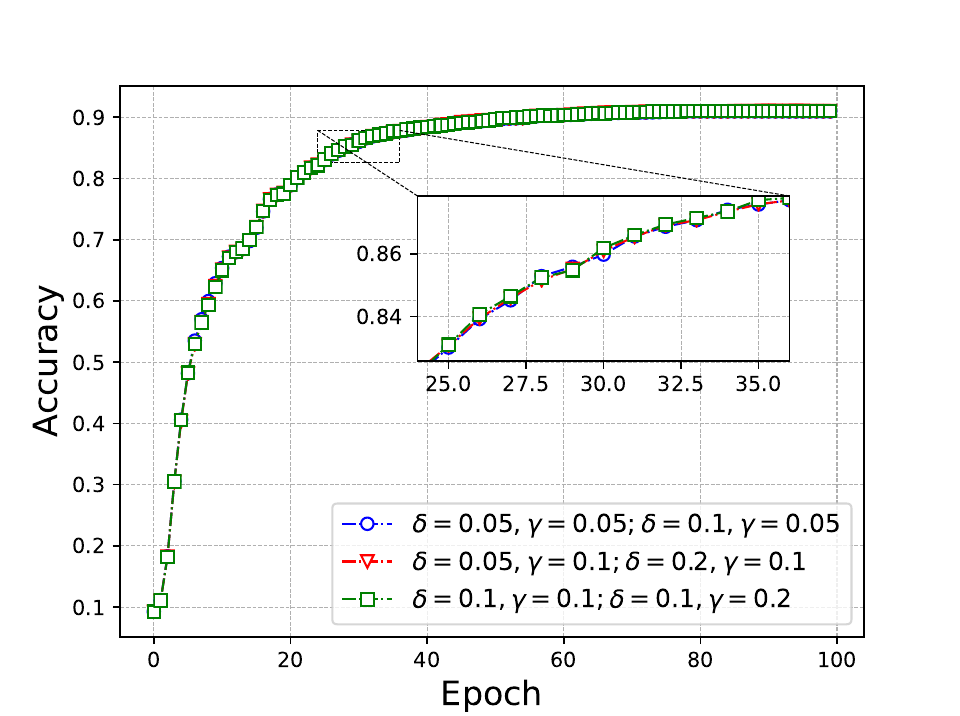}
        \caption{Single cluster}
    \end{subfigure}
    \begin{subfigure}{0.23\textwidth}
        \includegraphics[width=1\linewidth]{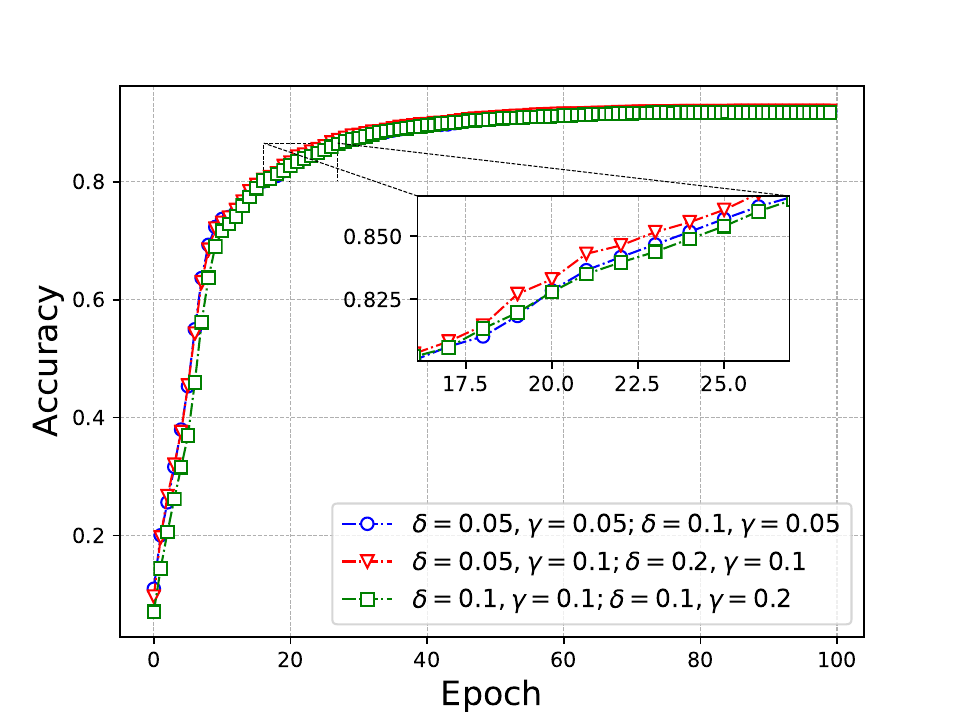}
        \caption{Multiple clusters}
    \end{subfigure}
    \caption{FL convergence over MNIST.}
    \label{fig:fl_acc_mnist}
\end{figure}

\begin{figure}[htbp!]
    \centering
    \begin{subfigure}{0.23\textwidth}
        \includegraphics[width=1\linewidth]{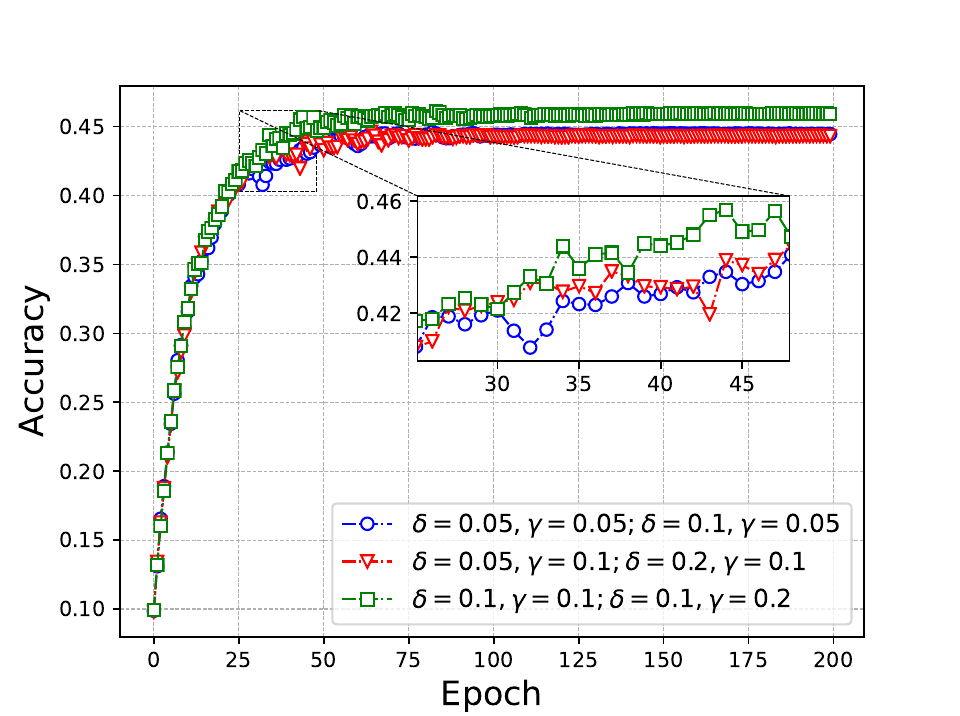}
        \caption{Single cluster}
    \end{subfigure}
    \begin{subfigure}{0.23\textwidth}
        \includegraphics[width=1\linewidth]{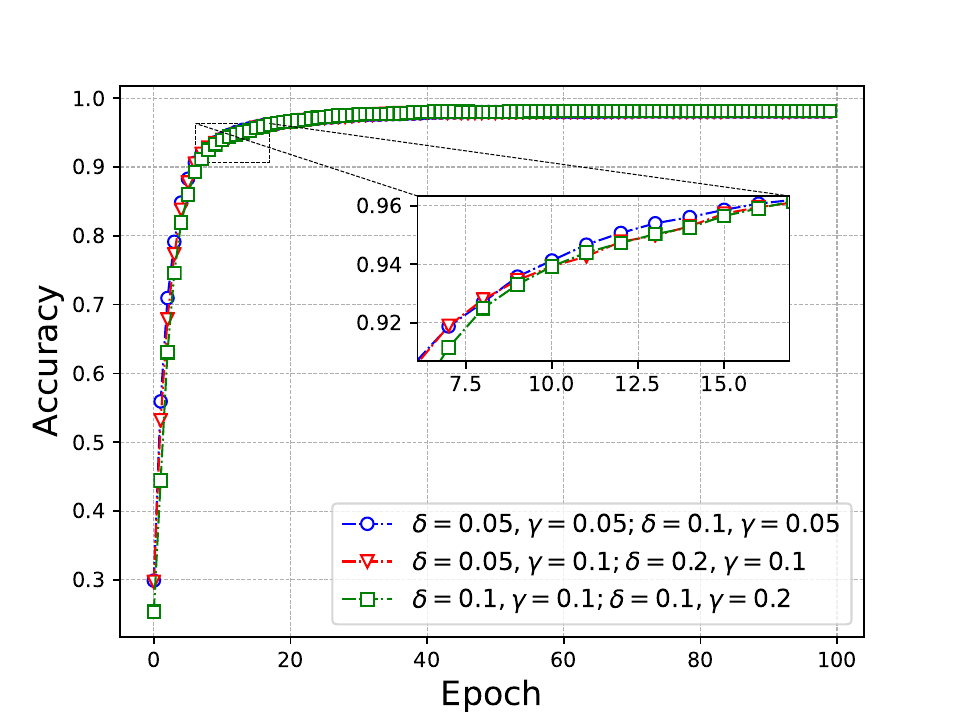}
        \caption{Multiple clusters}
    \end{subfigure}
    \caption{FL convergence over CIFAR-10.}
    \label{fig:fl_acc_cifar}
\end{figure}

\textbf{Convergence performance.} As previously noted, our proposed scheme facilitates unlearning through retraining, thereby necessitating standard FL training within each cluster. Figures \ref{fig:fl_acc_mnist} and \ref{fig:fl_acc_cifar} illustrate the impact of the parameters on cluster size, which in turn influences the convergence of FL training of both the global model within a single cluster and the ensemble results from multiple clusters after voting. We observe that larger clusters typically result in faster convergence but at the expense of achieved accuracy. This trend is particularly pronounced in more complex ML models compared to simpler models. These observations align consistently with the remarks made earlier.

\begin{figure}[htbp!]
    \centering
    \begin{subfigure}{0.23\textwidth}
        \includegraphics[width=1\linewidth]{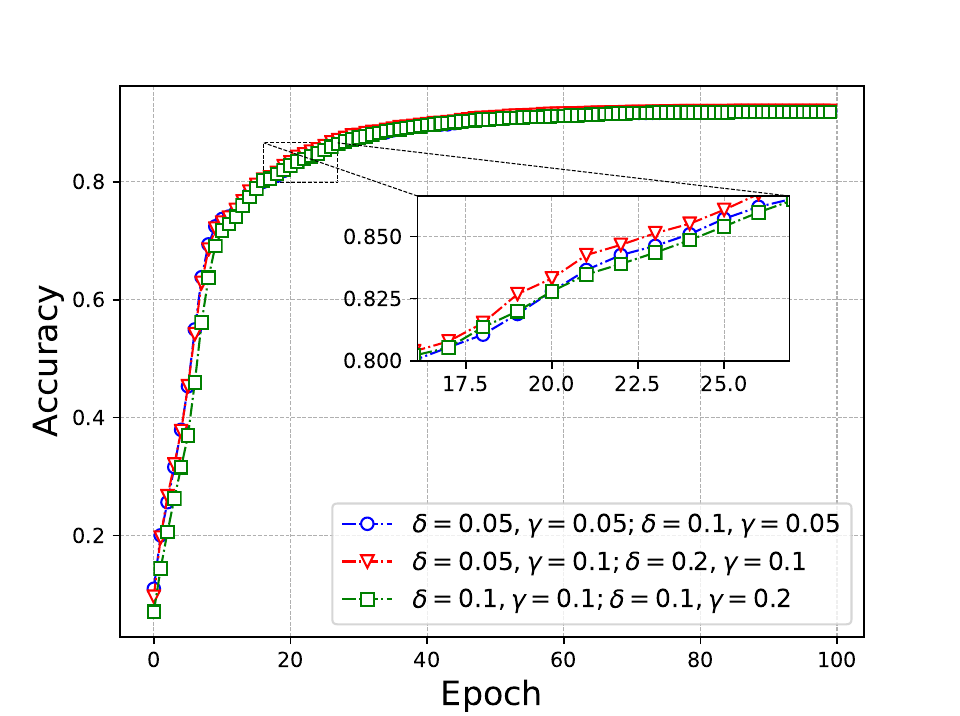}
        \caption{Sequential unlearning}
    \end{subfigure}
    \begin{subfigure}{0.23\textwidth}
        \includegraphics[width=1\linewidth]{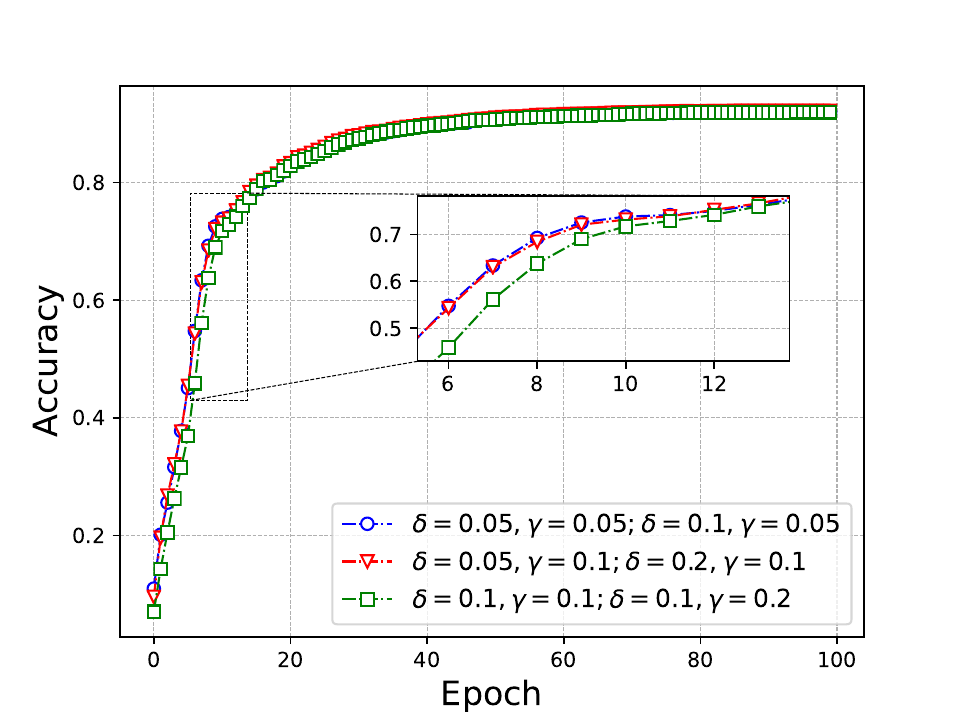}
        \caption{Batch unlearning}
    \end{subfigure}
    \caption{FU convergence over MNIST.}
    \label{fig:fu_acc_mnist}
\end{figure}

\begin{figure}[htbp!]
    \centering
    \begin{subfigure}{0.23\textwidth}
        \includegraphics[width=1\linewidth]{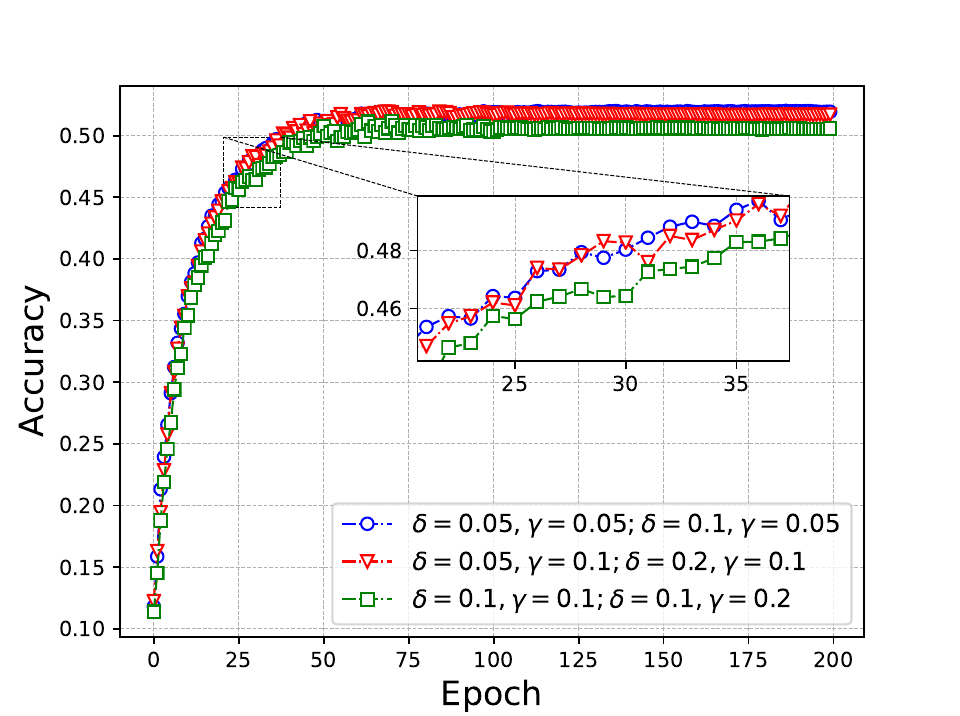}
        \caption{Sequential unlearning}
    \end{subfigure}
    \begin{subfigure}{0.23\textwidth}
        \includegraphics[width=1\linewidth]{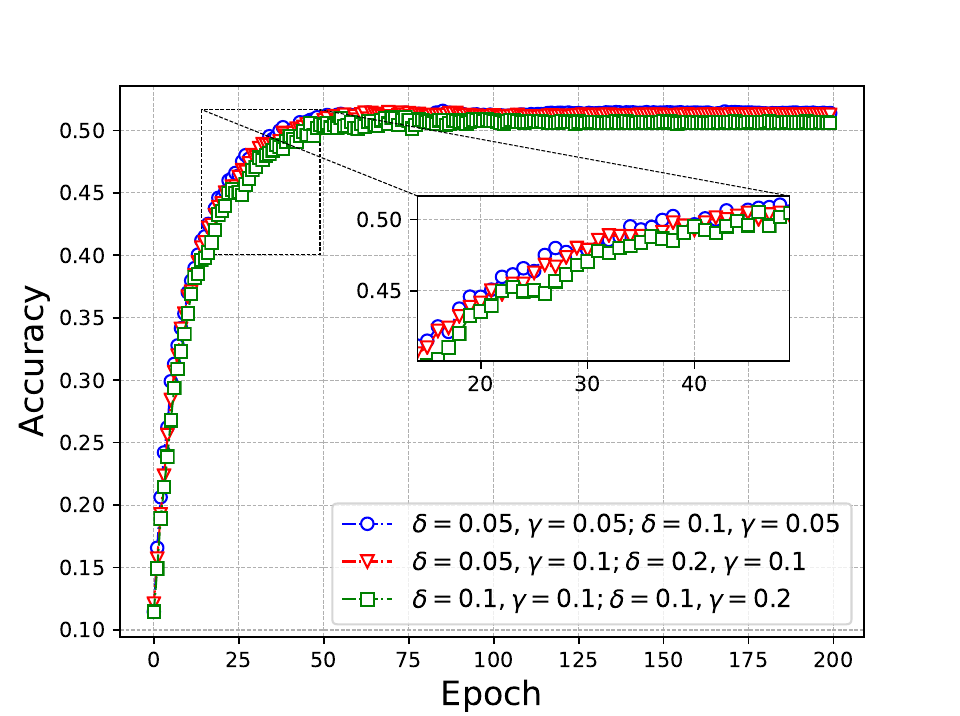}
        \caption{Batch unlearning}
    \end{subfigure}
    \caption{FU convergence over CIFAR-10.}
    \label{fig:fu_acc_cifar}
\end{figure}

Figures \ref{fig:fu_acc_mnist} and \ref{fig:fu_acc_cifar} present the convergence performance of our proposed scheme when handling sequential and batch requests during unlearning, respectively. It is observed that a larger value of $\delta$ and $\gamma$ typically leads to a lower accuracy of the unlearned model. This effect is more pronounced in complex ML models compared to simpler ones. This suggests that a clustering-based scheme, when facing potential dropout or adversarial users, may compromise unlearning performance to ensure security, regardless of whether it processes sequential or batch requests. Conversely, prioritizing unlearning efficiency could affect the scheme’s security.

\textbf{Aggregation efficiency.} Note that in the majority of settings, the number of users $N$ is not evenly divisible by the cluster size $k$. In such cases, our proposed scheme handles the remainder by randomly assigning them to clusters. As a result, smaller $N$ can sometimes result in larger clusters, leading to lower aggregation efficiency, i.e., increased runtime and higher communication costs, compared to those associated with slightly larger values of $N$. Examples illustrating this phenomenon are provided in Figure \ref{fig:agg_efficiency_mnist} and Figure \ref{fig:agg_efficiency_cifar}. Thus, careful selection of parameters is crucial when considering the remainder.

\begin{figure}[htbp!]
    \centering
    \begin{subfigure}{0.23\textwidth}
        \includegraphics[width=0.9\linewidth]{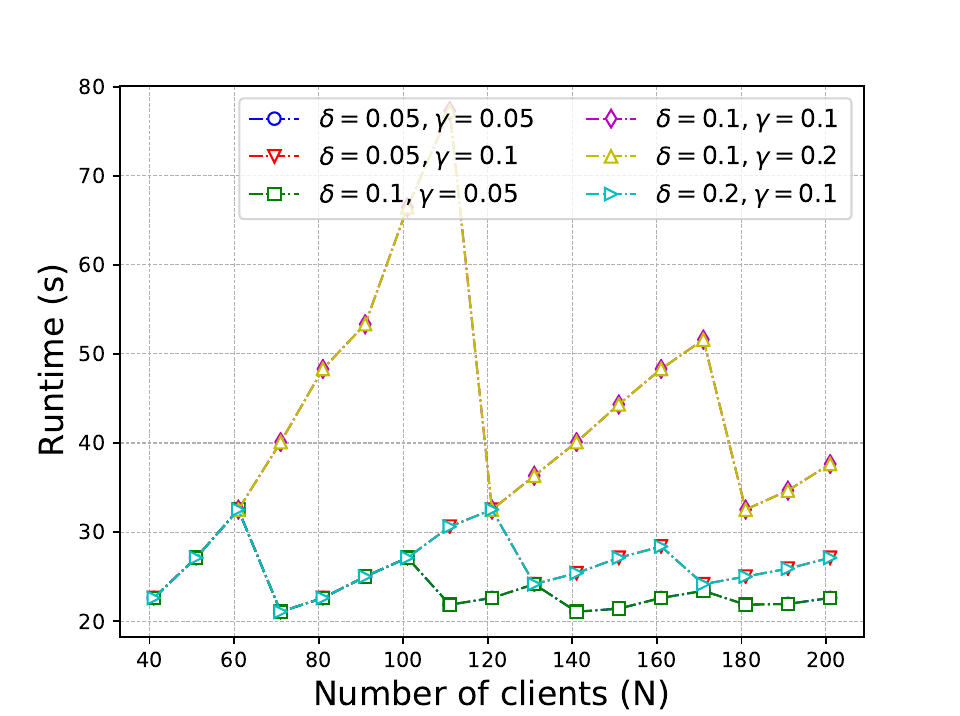}
        \caption{Runtime}
    \end{subfigure}
    \begin{subfigure}{0.23\textwidth}
        \includegraphics[width=0.9\linewidth]{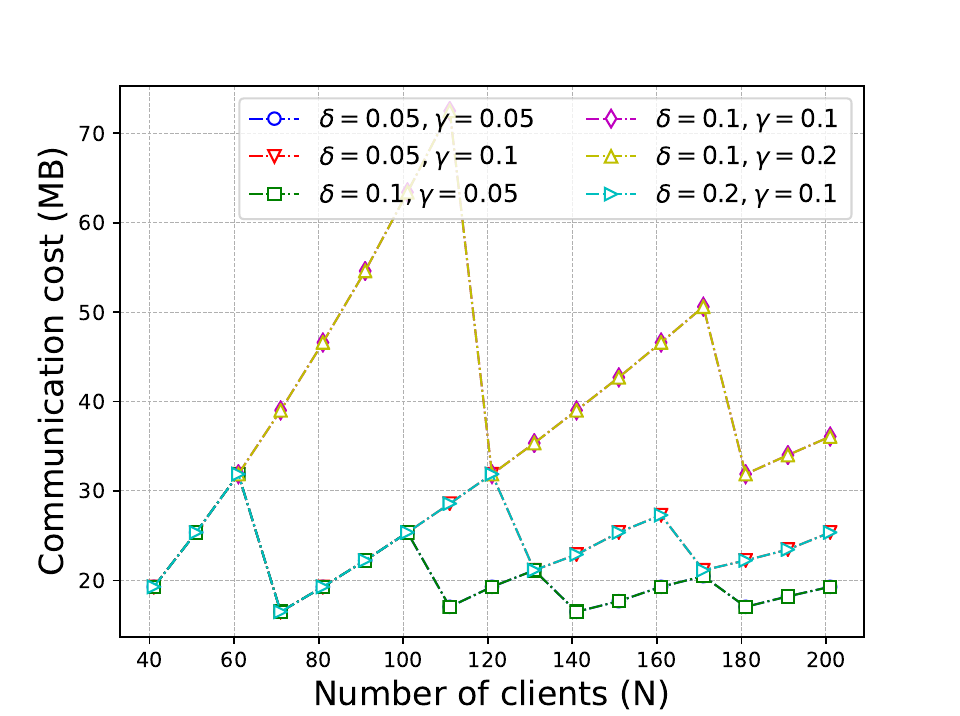}
        \caption{Communication cost}
    \end{subfigure}
    \caption{Aggregation efficiency over MNIST with the model $\mathcal{M}_1$ for one FL training or retraining round within each cluster.}
    \label{fig:agg_efficiency_mnist}
\end{figure}

\begin{figure}[htbp!]
    \centering
    \begin{subfigure}{0.23\textwidth}
        \includegraphics[width=0.9\linewidth]{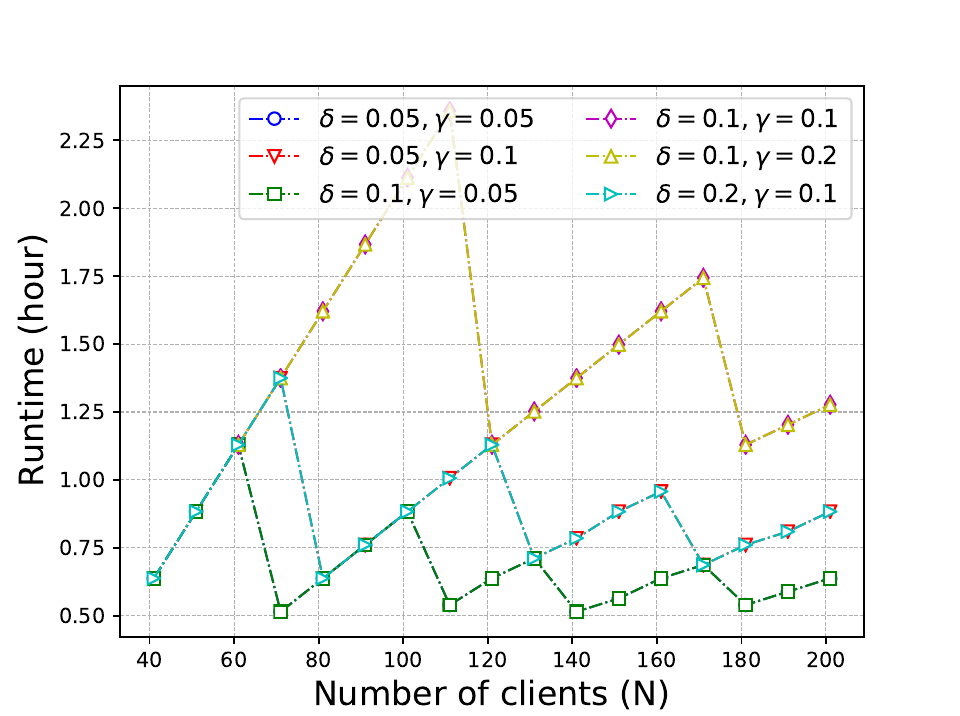}
        \caption{Runtime}
    \end{subfigure}
    \begin{subfigure}{0.23\textwidth}
        \includegraphics[width=0.9\linewidth]{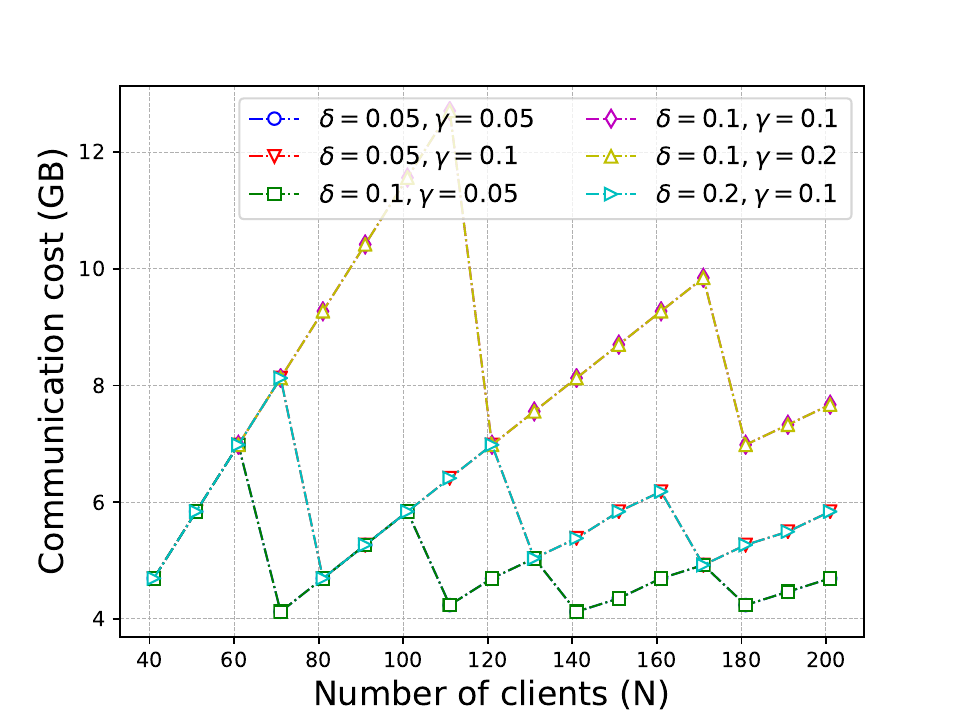}
        \caption{Communication cost}
    \end{subfigure}
    \caption{Aggregation efficiency over CIFAR-10 with the model $\mathcal{M}_2$ for one FL training or retraining round within each cluster.}
    \label{fig:agg_efficiency_cifar}
\end{figure}

Additionally, we offer the following remarks.

\begin{remark}
    (Clustering threshold) It is important to highlight that in scenarios where the number of users $N$ is relatively small, the minimal cluster size $k$ required to establish a $\sigma, \eta$-good clustering-based scheme may surpass the actual number of users $N$. For instance, this situation can be observed in settings where $N \leq 50$ with $\delta=0.1,\gamma=0.2$ as illustrated in Figure \ref{fig:cluster_size_delta_gamma}. This implies a threshold that our proposed scheme can only be effectively employed for clustering-based when there are a sufficient number of users available.
\end{remark}


\begin{remark}
    (Total runtime) The total runtime of retraining a cluster depends on both the number of epochs required for convergence, the costs associated with SecAgg+ protocols, and the scale of the deployed ML models. Larger clusters result in higher costs for privacy-preserving aggregation within the cluster but may lead to faster convergence, consequently requiring fewer epochs. A rough evaluation of this tradeoff can be obtained from the combination of the experimental results provided in Figures \ref{fig:fu_acc_mnist}, \ref{fig:fu_acc_cifar}, \ref{fig:agg_efficiency_mnist}, and \ref{fig:agg_efficiency_cifar}.
\end{remark}

\section{Related Works}
\label{sec:related_works}
This section offers an overview of recent FU studies, examines privacy risks that have been well-studied in FL systems but are often overlooked in FU systems and explores emerging security and privacy risks within MU schemes.

\textbf{Federated unlearning.} Compared to machine unlearning schemes, federated unlearning schemes involve multiple users, necessitating joint unlearning efforts that lead to significant communication and computation overhead among participants. Consequently, various studies have explored enhancements in communication and computation from different angles, including transferring computational tasks from users to the server \cite{cao2023fedrecover}, adopting periodic gradient aggregation instead of per-round aggregation \cite{tao2024communication}, employing clustered unlearning based on users' computational resource heterogeneity for asynchronous aggregation \cite{su2023asynchronous}, implementing importance-based data selection \cite{liu2024breaking}, utilizing more efficient optimizers for unlearning \cite{liu2022right}, dynamically selecting clients to reduce user participation \cite{lin2024incentive}, and opting for selective data revocation based on cost analysis \cite{ding2023strategic} and fairness \cite{su2024f2ul}.

\textbf{Privacy-Preserving Federated Unlearning.} However, it is worth noting that most current FU schemes primarily focus on boosting unlearning efficiency but overlook the potential information leakage from FL users' gradients, a privacy concern that has been extensively studied. Only a limited number of prior works take the privacy preservation of gradients into account. For instance, \cite{zhang2023fedrecovery} concentrates on protecting the privacy of the global model rather than users' data. \cite{tan2024unlink} investigates the privacy leakage in graph neural
networks. \cite{liu2024privacy} utilizes two servers to perform secure two-party computations. \cite{liu2021revfrf} is primarily focused on random forests. 

\textbf{Security and privacy risks in machine unlearning.} 
Beyond privacy risks from gradient information leakage in FU settings, several studies also explore various types of information leakage within MU systems. These include leakages from discrepancies between trained and unlearned models \cite{chen2021machine, hu2024learn, gao2022deletion, lu2022label}, as well as those arising from dependencies between adaptive unlearning requests \cite{chourasia2023forget}. Furthermore, from a security perspective, various studies demonstrate that adversarial users can submit crafted unlearning requests with untargeted goals, such as degrading the utility of the unlearned model \cite{hu2024duty,zhao2024static,ma2024releasing}, or targeted goals, such as injecting backdoors \cite{liu2024backdoor,di2022hidden,qian2023towards}. However, research on mitigation strategies for malicious unlearning and information leakage within unlearning systems is still in its early stages.

\section{Limitations and Future Works.}
In this section, we will discuss the limitations of our proposed FU scheme and outline potential research directions for future work.

\textbf{Limitations.} We note that our proposed scheme provides privacy guarantees only for gradients within each cluster. However, additional information leakage specific to unlearning settings may still exist in our scheme, as mentioned in Section \ref{sec:related_works}, which can be exploited through the differences between the original and unlearned versions of the model. In addition, in terms of bounding the number of unlearning requests, there is a lack of design for handling situations where the number of unlearning requests reaches the limit, as this can happen as unlearning requests accumulate over time. A straightforward solution is to re-cluster the users. However, SecAgg protocols must be re-initiated within the new clusters, making this approach time-consuming. Apart from security and privacy concerns, our approach is only compatible with clustering-based FU methods and cannot be integrated with other approximate FU methods that do not rely on a cluster structure, thus limiting its application scenarios.

\textbf{Future works.} As emerging privacy threats exist in unlearning systems. Hence, the first potential research direction is to explore techniques that can further protect against information leakage through model differences. This could be achieved using other privacy-enhancing techniques, such as differential privacy, to add perturbations not only to the gradients but also to the model parameters, or to combine additional cryptographic tools. However, considering various trade-offs, determining the optimal mechanism is not straightforward and warrants further exploration. In addition, investigating how to handle infinite unlearning requests is also worth exploring, as this is more practical for real-life unlearning applications. Designing such an FU system requires more consideration of dynamic participants, including not only dropout and unlearned users but also newly joined users. Last but not least, hybrid solutions need to be investigated to provide stronger security guarantees, improve unlearning efficiency, and enhance compatibility and scalability.

\section{Conclusions}
\label{sec:conclusions}

We systematically explored the integration of SecAgg+ protocols within the clustering-based federated unlearning scheme, to preserve privacy in scenarios with dynamic user participation. We developed a tailored clustering algorithm to ensure security guarantees alongside strategies for managing dynamic users. Additionally, we examined the trade-offs between privacy, unlearning performance, and efficiency under various parameter settings, demonstrating our approach's simplicity yet effectiveness in enhancing privacy beyond clustering-based FU schemes.



\bibliographystyle{IEEEtranS}
\bibliography{mybibliography}

\end{document}